\documentclass[a4paper]{amsproc}
\usepackage{amssymb}
\usepackage{amscd}
\usepackage{amsfonts}
\usepackage{amssymb}
\usepackage{cite}
\usepackage{amsmath}
\usepackage{epsfig}
\usepackage{graphicx, subfigure}
\usepackage{color, graphics}
\usepackage{pdfsync}
\usepackage{hyperref}
\usepackage[all,cmtip]{xy}
\usepackage{setspace}

\theoremstyle{plain}
 \newtheorem{thm}{Theorem}[section]
 \newtheorem{prop}{Proposition}[section]

\theoremstyle{definition}
 
 \newtheorem{dfn}{Definition}[section]
\theoremstyle{remark}
 \newtheorem{rem}{Remark}[section]
 \numberwithin{equation}{section}

\renewcommand{\le}{\leqslant}
\renewcommand{\ge}{\geqslant}\renewcommand{\geq}{\geqslant}

\allowdisplaybreaks

\newcommand{\R}{\mathbb{ R}}

\newcommand{\N}{\mathcal{ N}}

\newcommand{\A}{\mathbb{A}}

\DeclareMathOperator{\diag}{diag}

\DeclareMathOperator{\Ad}{Ad}
\DeclareMathOperator{\tr}{tr}

\def\diag{\mathrm{diag}}

\def\rank{\mathrm{rank}}
\def \i{{\rm i}}

\def\px1{p_{x_1}}
\def\px2{p_{x_2}}
\def\pu1{p_{u_1}}

\title[Gyroscopic Chaplygin systems]{Gyroscopic Chaplygin systems and integrable magnetic flows on spheres}



\author[Dragovi\'c, Gaji\'c, Jovanovi\'c]{\bfseries Vladimir Dragovi\'c, Borislav Gaji\'c, Bo\v zidar Jovanovi\'c}

\address{
Department of Mathematical Sciences, University of Texas at Dallas,  USA
\\
Mathematical Institute, Serbian Academy of Sciences and Arts, Belgrade, Serbia}
\email{Vladimir.Dragovic@utdallas.edu}

\address{
Mathematical Institute, Serbian Academy of Sciences and Arts, Belgrade, Serbia}
\email{gajab@mi.sanu.ac.rs}

\address{
Mathematical Institute, Serbian Academy of Sciences and Arts, Belgrade, Serbia}
\email{bozaj@mi.sanu.ac.rs}

\begin{document}

\begin{abstract}  We introduce and study the Chaplygin systems with gyroscopic forces. This natural class of nonholonomic systems has not been treated before. We put a special emphasis on the important subclass of such systems with magnetic forces.
The existence of an invariant measure and the problem of Hamiltonization are studied,  both within the Lagrangian and the almost-Hamiltonian framework. In addition, we introduce problems of rolling of a ball with the gyroscope without slipping and twisting over a plane and over a sphere in $\R^n$ as examples of gyroscopic $SO(n)$--Chaplygin systems. We describe an invariant measure and provide  examples of $SO(n-2)$--symmetric systems (ball with gyroscope) that
allow the Chaplygin Hamiltonization.  In the case of additional $SO(2)$--symmetry we prove that the obtained magnetic geodesic flows on the sphere $S^{n-1}$ are integrable.
In particular, we introduce the generalized Demchenko case in $\R^n$,
where the inertia operator of the system is proportional to the identity operator.
The reduced systems are automatically Hamiltonian and represent the magnetic geodesic flows on the spheres $S^{n-1}$ endowed with the round-sphere metric,  under the influence of a homogeneous magnetic field. The magnetic geodesic flow problem on the two-dimensional sphere is well known, but for $n>3$ was not studied  before. We perform explicit integrations in elliptic functions of the systems for $n=3$ and $n=4$, and provide the case study of the solutions in both situations.
\end{abstract}
\maketitle

\tableofcontents

\section{Introduction}

\subsection{Nonholonomic Lagrangian systems with gyroscopic forces}
The main aim of this paper is to introduce and study a general setting for
Chaplygin systems with gyroscopic forces, with a special  emphasis on the important subclass of the Chaplygin systems with magnetic forces. This class of nonholonomic systems, although quite natural, has not been treated before.

In his first PhD thesis, Vasilije Demchenko \cite{Dem1924, DGJ}, studied the rolling of a ball with a gyroscope without slipping over a sphere in  $\R^3$, by using the Voronec equations \cite{Vor1902, Vor1911, Vor1912}.
Inspired by this thesis, we consider the rolling of a ball with a gyroscope without slipping and twisting over a sphere in  $\R^n$.
This will provide us with examples of gyroscopic $SO(n)$--Chaplygin systems that reduce to integrable magnetic geodesic flows on a sphere $S^{n-1}$.

Let $(Q,\mathbf G)$ be a  Riemannian manifold.
 Consider a Lagrangian nonholonomic system $(Q,L_1,\mathcal D)$, where the constraints define a nonintegrable
distribution $\mathcal D$ on $Q$. The constraints are homogeneous and do not depend on time.
The Lagrangian, along with the difference of the kinetic and potential energy, contains an additional term, which is linear in velocities:
\[
L_1(q,\dot q)=\frac12(\mathbf G(\dot q),\dot q)+(\mathbf A,\dot q)-V(q).
\]
Here and throughout the text, $(\cdot,\cdot)$ denotes the parring between appropriate dual spaces,
while $\mathbf A$ is a one-form on $Q$. The metric $\mathbf G$ is also considered as a mapping $TQ\to T^*Q$.

A smooth path $q(t)\in Q,\; t\in\Delta$ is called
{\it admissible}  if  the velocity
$\dot q(t)$ belongs to ${\mathcal D}_{q(t)}$ for all $t\in\Delta$.
An admissible path $q(t)$ is a {\it
motion of the natural mechanical nonholonomic system} $(Q,L_1,\mathcal D)$ if it satisfies the Lagrange-d'Alembert equations
\begin{equation}\label{L1}
\delta L_1=\big(\frac{\partial L_1}{\partial q}-\frac{d}{dt}\frac{\partial L_1}{\partial \dot q},\delta q\big)=0, \quad \text{for all} \quad \delta q\in\mathcal D_q.
\end{equation}

The equations \eqref{L1} are equivalent to the equations
\begin{equation}\label{L}
\delta L=\big(\frac{\partial L}{\partial q}-\frac{d}{dt}\frac{\partial L}{\partial \dot q},\delta q\big)=\mathbf F(\dot q,\delta q), \quad \text{for all} \quad \delta q\in\mathcal D_q,
\end{equation}
where $L$ is the part of the Lagrangian $L_1$ which does not contain the term linear in velocities:
\begin{equation*}\
L(q,\dot q)=\frac12(\mathbf G(\dot q),\dot q)-V(q).
\end{equation*}
Here the additional force $\mathbf F(\dot q,\delta q)$  is defined as the exact two-form
\[
\mathbf F=d\mathbf A,
\]
where $\mathbf A$ is the one-form from the linear in velocities term of the Lagrangian $L_1$. We will subsequently consider a more general
class of systems where an additional force is given as a two-form which is neither exact nor even closed.

Systems with an additional force  defined by a closed  two-form $\mathbf F$ and without nonholonomic constraints are very well studied.
The corresponding Hamiltonian flows are usually called magnetic flows or twisted flows. For the problem of integrability
 of magnetic flows, see e.g.
 \cite{BK2017, BJ2008, T2016, MSY2008, S2002}.
Following tradition, we introduce

\begin{dfn}
Let $\mathbf F$ be a 2-form on $Q$. We refer to a system $(Q,L, \mathbf F, \mathcal D)$ as a \emph{natural mechanical nonholonimic system with gyroscopic forces}.
 The
additional gyroscopic force $\mathbf F(\dot q,\delta q)$ is called \emph{magnetic} if the form $\mathbf F$ is closed,
\[
d\mathbf F=0,
\]
and in this case we say that the system $(Q,L, \mathbf F, \mathcal D)$ is a \emph{natural mechanical nonholonomic system with a magnetic force}.
\end{dfn}

The equations of motion of a natural mechanical  nonholonomic system with a gyroscopic force $(Q,L, \mathbf F, \mathcal D)$ are given in \eqref{L}.

Starting from the notion of $G$--Chaplygin systems for nonholonomic systems without gyroscopic forces (see \cite{Baksa, St, Koi1992, BKMM, CCL2002, Naranjo2020}),
we introduce the following

\begin{dfn}
Assume that $Q$ is a principal bundle over $S$ with respect to a free action of a Lie group $G$, $\pi\colon Q\to S=Q/G$, and that $L$ and $\mathbf F$ are $G$--invariant.
Suppose that $\mathcal D$ is a principal connection, that is,  $\mathcal D$ is $G$--invariant, transverse to the orbits of the $G$--action,  and
$\rank\,\mathcal D=\dim S$.
Then we refer to $(Q,L,\mathcal D,G, \mathbf  F)$ as a \emph{gyroscopic $G$--Chaplygin system}.
\end{dfn}

Obviously, a gyroscopic $G$--Chaplygin system $(Q,L,\mathcal D,G, \mathbf  F)$ is $G$--invariant and reduces to the tangent space of the base-manifold $S=Q/G$.

\subsection{Outline and results of the paper}
In Section \ref{sec:gyroscopicfibredNOVA} we consider gyroscopic nonholonomic systems on fiber spaces. In Section
\ref{sec3} we employ them to describe a reduction procedure for the gyroscopic $G$--Chaplygin systems
(Theorem \ref{redukcijaSistema}).
The Chaplygin systems have a natural geometrical framework as connections on principal bundles (see \cite{Koi1992}). On the other hand,  nonholonomic systems were incorporated into the geometrical framework
of the Ehresmann connections on fiber spaces in \cite{BKMM}. In this paper, we combine the approach of \cite{BKMM} with the Voronec nonholonomic equations, see  \cite{Vor1902}.

In Section \ref{sec4} we derive the equations of motion of the reduced gyroscopic $G$--Chaplygin systems in an almost-Hamiltonian form and study the existence of an invariant measure
(Theorem \ref{novamera}).  A closely related problem is
the Hamiltonization of nonholonomic systems
(see \cite{Chaplygin1911, St, BM, BM2008, BN, BBM, BMT, BBM2013, EKMR2005, CCL2002, FJ2004, Jov2019, Jov2018b}).
In Section \ref{sec5} we consider the Chaplygin reducing multiplier and the time reparametrization  of magnetic Chaplygin systems, both within the Lagrangian and the Hamiltonian framework (see Theorem \ref{opsta}).

In Section \ref{sec6} we briefly review the results about integrable nonholonomic problems of rolling of a ball with the
gyroscope, without slipping and twisting, over a plane and over a sphere in the three-dimensional space. In particular, we present the Demchenko integrable case \cite{Dem1924} and the Zhukovskiy condition for the system  \cite{Zhuk1893}.

In Section \ref{7} we introduce the problems of rolling of a ball with a
gyroscope, without slipping and twisting, over a plane and over a sphere in $\R^n$.
We describe the reduction (Propositions \ref{redukcija} and \ref{redukcija2}) and an invariant measure (Proposition \ref{invarijantnaMera})  of these new systems.
The obtained systems are examples of gyroscopic $SO(n)$--Chaplygin systems that reduce to magnetic flows.

In Section \ref{sec8} we provide  examples of  $SO(n-2)$--symmetric systems (ball with gyroscope) that
allow the Chaplygin Hamiltonization (Theorem \ref{hamiltonizacija}).
We also prove the integrability of the obtained magnetic geodesic flows on a sphere in $\R^n, \, n\ge3$ in the case of $SO(2)\times SO(n-2)$--symmetry (Theorem \ref{integrabilni}).
Note that the phase space of a nonholonomic system that is integrable after the Chaplygin Hamiltonization is foliated by $d$-dimensional invariant tori,
where the system is subject to a non-uniform
quasi-periodic motion of the form
\begin{equation}
\dot\varphi_1=\omega_1/\Phi(\varphi_1,\dots,\varphi_d),\dots,
\dot\varphi_d=\omega_d/\Phi(\varphi_1,\dots,\varphi_d), \qquad \Phi>0,
\label{Jacobi}
\end{equation}
with some $d, \, d\le n$.
In Theorem \ref{integrabilni} we present two examples of such systems, one with  $d=2$ and $n=3$ and another one with $d=3$ and any $n>3$.

Finally, in Section \ref{sec9} we consider  the case when the inertia operator for systems is $SO(n)$--invariant, i.e. it satisfies the Zhukovskiy condition in $\R^n$ with an additional non-twisting condition. We will refer to such  systems as {\it the generalized Demchenko case without twisting in $\R^n$}. The reduced systems are automatically Hamiltonian. They represent the magnetic geodesic flow on a sphere $S^{n-1}$ endowed with the round-sphere metric,  under a influence of the homogeneous magnetic field placed in the ambient space $\R^n$.
The magnetic geodesic flow problem on a two-dimensional sphere is well known (see  \cite{S2002}).
However, the magnetic geodesic flow problems for $n>3$ have not been studied before.
We prove the complete integrability of the system on the three-dimensional sphere (Theorem \ref{3sfera}). We conclude the paper with a detailed analysis of the motion of the generalized Demchenko systems without twisting for $n=3$ and $n=4$ in terms of elliptic functions.

\section{Nonholonomic systems with gyroscopic forces on fibred spaces}\label{sec:gyroscopicfibredNOVA}

\subsection{The Voronec equations}\label{VorPr}

Following Demchenko\footnote{Demchenko's PhD advisor, Anton Bilimovi\'c (1879-1970), was a distinguished student of Peter Vasilievich
Voronec (1871-1923) and one of the founders of Belgrade's Mathematical Institute. We note that some recent results (see \cite{BT2020, BMT2020}) are
inspired by Bilimovi\' c's work in nonholonomic mechanics \cite{Bilimovic1913a, Bilimovic1913b, Bilimovic1914, Bilimovic1915, Bilim}.}
\cite{Dem1924, DGJ}, we  recall the Voronec equations for nonholonomic systems \cite{Vor1902}.
We will then employ them to formulate the reduced equations  of gyroscopic Chaplygin systems.
Here we assume that the constraints may be time-dependent and nonhomogeneous.

Let $q=(q_1,\dots,q_{n+k})$ be local coordinates of the configuration space $Q$.
Consider a nonholonomic system with kinetic energy $T=T(t,q,\dot q)$, generalized forces $Q_s=Q_s(t,q,\dot q)$  that correspond to
coordinates $q_s$, and time-dependent nonhomogeneous nonholonomic constraints
\begin{equation}\label{v1}
\dot q_{n+\nu}=\sum_{i=1}^n a_{\nu i} (q,t) \dot q_i + a_\nu (q,t),  \qquad \nu=1,2,\dots,k.
\end{equation}

Let $T_c$ be the kinetic energy $T$ after imposing the constraints \eqref{v1}. Let $K_{\nu}$ be the partial derivatives of the kinetic energy $T$
with respect to $\dot q_\nu$, $\nu=1,2,\dots,k$, restricted to the constrained subspace. We assume that the constraints \eqref{v1} are imposed after the differentiation and get:
\begin{align*}
& T_c(t,q_1,\dots,q_{n+k},\dot q_1,\dots,\dot q_n)=T(t,q,\dot q)\vert_{\dot q_{n+\nu}=\sum_{i=1}^n a_{\nu i} (q,t) \dot q_i + a_\nu (q,t)},  \\
& K_\nu(t,q_1,\dots,q_{n+k},\dot q_1,\dots,\dot q_n)=\frac{\partial T}{\partial \dot q_{n+\nu}}(t,q,\dot q)\vert_{\dot q_{n+\nu}=\sum_{i=1}^n a_{\nu i} (q,t) \dot q_i + a_\nu (q,t)}.
\end{align*}

The equations of motion of the given noholonomic system can be presented in a  form which does not use
the Lagrange multipliers:
\begin{align}
& \frac{d}{dt}\frac{\partial T_c}{\partial \dot q_i}=\frac{\partial T_c}{\partial q_i} + Q_i+\sum_{\nu=1}^k a_{\nu i}\big(\frac{\partial T_c}{\partial q_{n+\nu}}+Q_{n+\nu}\big)+ \sum_{\nu=1}^k K_\nu\big(\sum_{j=1}^n A_{ij}^{(\nu)} \dot q_j+A_j^{(\nu)}\big). \label{v4}
\end{align}
The derivation of these equations is based on the Lagrange-d'Alembert principle and follows Voronec \cite{Vor1902}.
Here $i=1,\dots,n$. The components $A_{ij}^{(\nu)}$ and $A_i^{(\nu)}$ are functions of the time $t$ and the coordinates $q_1,\dots,q_{n+k}$ given by
\begin{align*}
& A_{ij}^{(\nu)}=  \big(\frac{\partial a_{\nu i}}{\partial q_j}+\sum_{\mu=1}^k a_{\mu j}\frac{\partial a_{\nu i}}{\partial q_{n+\mu}}\big)
                   -\big(\frac{\partial a_{\nu j}}{\partial q_i}+\sum_{\mu=1}^k a_{\mu i}\frac{\partial a_{\nu j}}{\partial q_{n+\mu}}\big), \\
&  A_i^{(\nu)}= \big(\frac{\partial a_{\nu i}}{\partial t}+\sum_{\mu=1}^k a_{\mu }\frac{\partial a_{\nu i}}{\partial q_{n+\mu}}\big)
                   -\big(\frac{\partial a_{\nu}}{\partial q_i}+\sum_{\mu=1}^k a_{\mu i}\frac{\partial a_{\nu}}{\partial q_{n+\mu}}\big).
\end{align*}

When all considered objects do not depend on the variables $q_{n+\nu}, \, \nu=1,2,\dots,k$, we have  a \emph{Chaplygin system}. Then the equations \eqref{v4}
are called the  \emph{Chaplygin equations}. The Voronec and the Chaplygin equations, along with the
equations of  nonholonomic systems written in terms of quasi-velocities, known as \emph{the Euler-Poincar\'e-Chetayev-Hamel} equations,
 form  core tools in the study of nonholonomic mechanics (see \cite{NeFu, BKMM, deLeon, EKMR2005, EK2019, Zenkov2016}).

\subsection{The Ehresmann connections and systems with gyroscopic forces}\label{sec:gyroscopicfibred}
Consider a natural mechanical nonholonimic system with a gyroscopic  force $(Q,L, \mathbf F, \mathcal D)$.
After Bloch, Krishnaprasad, Marsden, and Murray \cite{BKMM},
we assume that $Q$ has a structure of a fiber bundle $\pi\colon Q\to S$ over a base manifold $S$ and that the distribution $\mathcal D$ is transverse to the fibers
of $\pi$:
\[
T_q Q=\mathcal D_q \oplus \mathcal V_q, \qquad \mathcal V_q=\ker d\pi(q).
\]
The space $\mathcal V_q$ is called the  the \emph{vertical space} at $q$.
The distribution $\mathcal D$ can be seen as the kernel of
a vector-valued one-form $A$ on $Q$, which defines the \emph{Ehresmann connection}, that satisfies

\begin{itemize}

\item[(i)] $A_q\colon T_q Q \to\mathcal V_q$ is a linear mapping, $q\in Q$;

\item[(ii)] $A$ is a projection: $A(X_q)=X_q$, for all $X_q\in\mathcal V_q$.

\end{itemize}

The distribution $\mathcal D$ is called \emph{the horizontal space} of the Ehresmann connection $A$.
By $X^h$ and $X^v$ we denote the horizontal and the vertical component of the vector field $X\in\mathfrak{X}(Q)$.
The \emph{curvature} $B$ of the connection $A$ is a vertical vector-valued two-form defined by
\[
B(X,Y)=-A([X^h,Y^h]).
\]

Let $\dim Q=n+k$ and $\dim S=n$.
There exist local ``adapted" coordinates $q=(q_1,\dots,q_{n+k})$ on $Q$, such that the projection $\pi\colon Q\to S$ and the constraints defining $\mathcal D$ are given by
\begin{align*}
& \pi\colon (q_1,\dots,q_n,q_{n+1},\dots,q_{n+k}) \longmapsto (q_1,\dots,q_n),\\
& \dot q_{n+\nu}=\sum_{i=1}^n a_{\nu i} (q) \dot q_i ,  \qquad\qquad \nu=1,\dots,k.
\end{align*}
Here $(q_1,\dots,q_n)$ are the local coordinates on $S$.
Then, locally, we also have
\begin{align*}
& A=\sum_{\nu=1}^k \omega^\nu \frac{\partial}{\partial q_{n+\nu}}, \quad \omega^\nu=dq_{n+\nu}-\sum_{i=1}^n a_{\nu i} dq_i, \\
& X^h=\big(\sum_{l=1}^{n+k}X_l\frac{\partial}{\partial q_{l}}\big)^h=\sum_{i=1}^{n}X_i\frac{\partial}{\partial q_{i}}
+\sum_{\nu=1}^k\sum_{i=1}^{n}a_{\nu i}X_i\frac{\partial}{\partial q_{n+\nu}},\\
& X^v=\big(\sum_{l=1}^{n+k}X_l\frac{\partial}{\partial q_{l}}\big)^v=
\sum_{\nu=1}^k\big(X_{n+\nu}-\sum_{i=1}^{n}a_{\nu i}X_i\big)\frac{\partial}{\partial q_{n+\nu}},\\
& B(\frac{\partial}{\partial q_{i}},\frac{\partial}{\partial q_{j}})=\sum_{\nu=1}^k B^\nu_{ij}\frac{\partial}{\partial q_{n+\nu}},\\
& \mathbf F=\sum_{1\le s<l \le n+k} F_{sl} dq_s\wedge dq_l.
\end{align*}
Here $B_{ij}^\nu(q)=A^{(\nu)}_{ij}(q)$, where  $A^{(\nu)}_{ij}(q)$ come
from the Voronec equations \eqref{v4} with homogeneous constraints, which do not depend on time.
The generalized forces  $Q_s=Q_s(q,\dot q)$,  $s=1,\dots,n+k$ are the sums of the potential and the gyroscopic forces
\[
Q_s=Q_s^V+Q_s^\mathbf F, \quad Q_s^V=-{\partial V}/{\partial q_s}, \quad Q_s^\mathbf F=\sum_{l=1}^{n+k} F_{sl} \dot q_l.
\]
The Voronec equations \eqref{v4} take the form:
\begin{equation}
\frac{d}{dt}\frac{\partial L_c}{\partial \dot q_i}=\frac{\partial L_c}{\partial q_i} +\sum_{\nu=1}^k a_{\nu i}\frac{\partial L_c}{\partial q_{n+\nu}}
+\sum_{\nu=1}^k\sum_{j=1}^n\frac{\partial L}{\partial\dot q_{n+\nu}} B_{ij}^{\nu} \dot q_j+Q_i^\mathbf F+\sum_{\nu=1}^k a_{\nu i}Q_{n+\nu}^\mathbf F,  \label{v4*}
\end{equation}
$(i=1,\dots,n)$, where $L_c$ is the constrained Lagrangian
$L_c=L(q,\dot q^h)=T_c-V$.
In a compact form, the equations can be expressed as:\footnote{One can compare the form of equations \eqref{v5*} with the compact
form of the  Voronec equations
 obtained from the Voronec principle, see e.g. \cite{DGJ}.}
\begin{equation}\label{v5*}
\delta L_c=\mathbb FL(q,\dot q)( B(\dot q,\delta q))+\mathbf F(\dot q,\delta q)
\end{equation}
for all virtual displacements
\[
\delta q=\sum_{s=1}^{n+k} \delta q_s\frac{\partial}{\partial q_s}\in \mathcal D_q.
\]
Here $\delta L_c$ is the variational derivative of the constrained Lagrangian along the variation $\delta q$ and $\mathbb FL$ is the fiber derivative of $L$:
\begin{align*}
&\delta L_c=\big(\frac{\partial L_c}{\partial q}-\frac{d}{dt}\frac{\partial L_c}{\partial \dot q},\delta q\big)=
\sum_{s=1}^{n+k}\big(\frac{\partial L_c}{\partial q_s}-\frac{d}{dt}\frac{\partial L_c}{\partial \dot q_s}\big)\delta q_s,\\
&\mathbb FL(q,X)(Y)=\frac{d}{ds}\vert_{s=0}L(q,X+sY), \quad X,Y \in T_q Q,    \\
&\mathbb FL(q,\dot q)(B(\dot q,\delta q))=\sum_{\nu=1}^k \frac{\partial L}{\partial \dot q_{n+\nu}}(q,\dot q)B^\nu(\dot q,\delta q).
\end{align*}
See \cite{BKMM} for the case without gyroscopic two-form $\mathbf F$.

Note that,
even in the case when the two form $\mathbf F$ is exact $\mathbf F=d\mathbf A$,
it is convenient to use the Lagrangian $L$ and the form of the equations \eqref{v5*}, rather then the Lagrangian $L_1$
with the term linear in velocities.

\begin{rem}
In the case when the constraints are nonhomogeneous and time dependent \eqref{v1}, the coefficients $A^{(\nu)}_{ij}$, $A^{(\nu)}_i$ can be also
interpreted as the components of the curvature of the Ehresmann connection of the fiber bundle $\pi: Q\times \mathbb R\to S\times \mathbb R$
(see Bak\v sa \cite{Baksa2012}).
\end{rem}

\section{The Gyroscopic Chaplygin systems}\label{sec3}

In addition to the assumptions from  Subsection \ref{sec:gyroscopicfibred}, we now assume that the fibration $\pi: Q\to S$ is determined by a free action of a $k$--dimensional
Lie group $G$ on $Q$, so that $S=Q/G$ and that the constraint distribution $\mathcal D$,
the gyroscopic two-form $\mathbf F$ and
the Lagrangian $L=T-V$ are $G$--invariant. Then $A$ is a principal connection and
the nonholonomic system \eqref{v5*} is $G$--invariant and reduces to the tangent bundle of the base manifold $S$
by the identification $TS=\mathcal D/G$. More precisely,  we use the following definition.

\begin{dfn}
Let $\mathbf G$, $V$, and $\mathbf  F$ be a $G$--invariant metric, a potential and a two-form on $Q$. The \emph{reduced metric} $\mathbf g$, the \emph{reduced potential} $v$,
and the \emph{reduced two-form} $\mathbf  f$ on $S$ are defined by:
\[
\mathbf g(X,Y)\vert_x=\mathbf G(X^h,Y^h)\vert_q,  \qquad v(x)=V(q), \qquad \mathbf f(X,Y)\vert_x=\mathbf F(X^h,Y^h)\vert_q.
\]
Here $X^h,Y^h$ are \emph{the horizontal lifts of $X,Y$ at a point $q\in\pi^{-1}(x)$} defined by
\[
d\pi\vert_q(X^h)=X, \qquad d\pi\vert_q(Y^h)=Y,  \qquad X^h,Y^h\in\mathcal D_q.
\]
\end{dfn}

Note that we do not impose any additional assumptions on $\mathbf F$. In particular, $\mathbf F$ does not need to be of the form $\mathbf F=\pi^*\mathbf w$, where $\mathbf w$ is a 2-form on the base manifold $S$.

The equations \eqref{v5*} are $G$--invariant and they reduce to $TS$
\begin{equation}
\delta l=\big(\frac{\partial l}{\partial x}-\frac{d}{dt}\frac{\partial l}{\partial
\dot x},\delta x\big)= \mathbf{JK}(\dot x,\delta x)+\mathbf  f(\dot x,\delta x)\quad
\text{for all} \quad \delta x\in T_x S,
\label{ChaplyginRed}
\end{equation}
where
\[
l=\frac12(\mathbf g(\dot x),\dot x)-v(x)
\]
is the reduced Lagrangian
and the term\footnote{Let us note that in \cite{EKMR2005}, the term ``JK" is used for the associated semi-basic two-form $\sigma$ on $T^*S$ given below.
} $\mathbf{JK}(\cdot,\cdot)$ depends on the metric and the curvature of the connection, induced by $\mathbb FL( B(\cdot,\cdot))$.
The term $\mathbf{JK}(\cdot,\cdot)$ can be described as follows. Consider the (0,3)-tensor field $\Sigma$ on $S$ defined by
\begin{equation}\label{eq:Sigma}
\Sigma(X,Y,Z)\vert_x=\mathbb FL(q,X^h)(B(Y^h,Z^h))\vert_q, \qquad q\in\pi^{-1}(x),
\end{equation}
where $X^h,Y^h,Z^h$ are the horizontal lifts of vector fields $X,Y,Z$  on $S$.
Then $\Sigma$ is skew-symmetric with respect to the
second and the third argument, and
\begin{equation}\label{JKsigma}
\mathbf{JK}(X,Y)\vert_{(x,\dot x)}=\Sigma(\dot x,X,Y).
\end{equation}

\begin{rem}\label{JKclan}
Let us explain the notation for the {\bf JK}-term. It is obtained from the
natural paring of the momentum mapping of the $G$-action $J: TQ\to\mathfrak g^*$ and the curvature $K: TQ\times TQ\to \mathfrak g$
of the principal connection $A$, where
$\mathfrak g$ is the Lie algebra of the Lie group $G$. Namely,
we have a canonical identification of the vertical space $\mathcal V_q$ with the Lie algebra $\mathfrak g$.
Then the curvature of the Ehresmann connection $B$ is $\mathfrak g$--valued and coincides
with the curvature $K$ of the principal connection. Also, within this identification, the fiber derivative $\mathbb FL(q,\dot q)$ in the
direction of the vertical vector $\xi\in\mathfrak g \cong\mathcal V_q$  becomes the value of the momentum mapping $J$ of the $G$-action
evaluated at $\xi$. In this way the expression \eqref{eq:Sigma},
as the natural paring of the tangent bundle momentum mapping $J$ and the curvature two-form $K$, defines a (0,3)-tensor field $\Sigma$ on $S$.
On the other hand, the {\bf JK}-term defined by \eqref{JKsigma} is a semi-basic 2-form on $TS$.
\end{rem}

\begin{dfn}
We refer to $(S,l,\mathbf{JK},\mathbf f)$ as \emph{a reduced gyroscopic  $G$--Chaplygin system}.  In the case when $\mathbf  f$ is a closed form, we call it
\emph{a reduced magnetic $G$--Chaplygin system}.
\end{dfn}
The equations of motion of the reduced gyroscopic  $G$--Chaplygin system $(S,l,\mathbf{JK},\mathbf f)$ are described in \eqref{ChaplyginRed}.

We summarize the above considerations in the following statement.

\begin{thm}\label{redukcijaSistema}
The solutions of  the gyroscopic $G$--Chaplygin system $(Q,L,\mathcal D,G, \mathbf  F)$ project to  solutions of the reduced gyroscopic $G$--Chaplygin system
$(S,l,\mathbf{JK},\mathbf f)$. Let $x(t)$ be a solution of the reduced system \eqref{ChaplyginRed} with the initial conditions $x(0)=x_0$, $\dot x(0)=X_0\in T_{x_0} S$
and let $q_0\in \pi^{-1}(x_0)$. Then the horizontal lift $q(t)$ of $x(t)$ through $q_0$ is the solution of the original system \eqref{L}, i.e., \eqref{v5*},  with the initial conditions $q(0)=x_0$, $\dot q(0)=X^h_0\in \mathcal D_{q_0}$.
\end{thm}

\begin{rem}\label{primedba}
If $\mathbf f$ is an exact magnetic form, e.g. $\mathbf f=d\mathbf a$,
 then the equations \eqref{ChaplyginRed} are equivalent to
\begin{equation}
\delta l_1=\big(\frac{\partial l_1}{\partial x}-\frac{d}{dt}\frac{\partial l_1}{\partial
\dot x}, \delta x\big)= \mathbf{JK}(\dot x,\delta x)\qquad
\text{for all} \qquad \delta x\in T_x S,
\label{ChaplyginRed2}
\end{equation}
where the Lagrangian $l_1$, given by
\[
l_1=\frac12(\mathbf g(\dot x),\dot x)+(\mathbf a,\dot x)-v(x),
\]
has the linear term $(\mathbf a,\dot x)$.
\end{rem}

\begin{rem}\label{giroskopskiTenzor}
Within the affine connection approach to the Chaplygin reduction, it is convenient to introduce (1,2)-tensor fields $\mathbf B$ and $\mathbf C$ defined
by (see Koiller \cite{Koi1992} and Cantrijn Cantrijn, Cortes, de Leon, and  Martin de Diego  \cite{CCL2002})
\[
\Sigma(X,Y,Z)=\mathbf g(\mathbf B(X,Y),Z)=\mathbf g(X,\mathbf C(Y,Z)).
\]
In \cite{GajJov2019}, the tensor field $\mathbf B$ was used,
while here we work with the skew-symmetric tensor $\mathbf C$.
Note that $\mathbf C$ is equal to the negative gyroscopic tensor $\mathcal T$ defined by Garcia-Naranjo \cite{Naranjo2019a, Naranjo2019b}.
\end{rem}

Note that if $\mathbf F$ is magnetic, then $\mathbf f$ is not necessarily magnetic. Indeed, we have

\begin{prop}
Assume that the form $\mathbf F$ is closed. Then the reduced form $\mathbf f$ is closed if and only if
\begin{equation}\label{uslovDinv}
\mathbf F([X^h,Y^h]-[X,Y]^h,Z^h)+\mathbf F([Z^h,X^h]-[Z,X]^h,Y^h)+\mathbf F([Y^h,Z^h]-[Y,Z]^h,X^h)=0,
\end{equation}
for all vector fields $X$, $Y$, $Z$ on $S$.
In the adapted coordinates $q=(q_1,\dots,q_{n+k})$ on $Q$ described in Subsection \ref{sec:gyroscopicfibred}, the condition \eqref{uslovDinv} is equivalent to the equations
\begin{equation}\label{uslovD}
\sum_{\nu=1}^k \big( B^\nu_{ij}F_{p,n+\nu}+B^\nu_{pi}F_{j,n+\nu}+B^\nu_{jp}F_{i,n+\nu}\big)=0, \qquad 1\le i,j,p \le n.
\end{equation}
In particular, if the curvature $B$ of the Ehresmann connection vanishes (equivalently, the curvature $K$ of the principal connection vanishes),
then $\mathbf f$ is closed.
\end{prop}

\begin{proof}
Since $\mathbf F$ is magnetic, we have
\begin{align*}
d\mathbf  F(X',Y',Z')
=& X' \mathbf F(Y',Z')+ Y' \mathbf F(Z',X')+Z'\mathbf F(X',Y')\\
&-\mathbf F([X',Y'],Z')-\mathbf F([Z',X'],Y')-\mathbf F([Y',Z'],X')=0,
\end{align*}
for arbitrary vector fields $X',Y',Z'$ on $Q$. On the other hand, by using the above relation and the definition of $\mathbf f$
that depends on the horizontal distribution $\mathcal D$, we get
\begin{align*}
d\mathbf  f(X,Y,Z)\vert_x
=& \big(X \mathbf f(Y,Z)+ Y \mathbf f(Z,X)+Z\mathbf f(X,Y)\\
&-\mathbf   f([X,Y],Z)-\mathbf f([Z,X],Y)-\mathbf f([Y,Z],X)\big)\vert_x\\
=& \big(X^h \mathbf F(Y^h,Z^h)+ Y^h \mathbf F(Z^h,X^h)+Z^h\mathbf F(X^h,Y^h)\\
&-\mathbf F([X,Y]^h,Z^h)-\mathbf F([Z,X]^h,Y^h)-\mathbf F([Y,Z]^h,X^h)\big)\vert_q \\
=& \big(\mathbf F([X^h,Y^h],Z^h)+\mathbf F([Z^h,X^h],Y^h)+\mathbf F([Y^h,Z^h],X^h)\\
&-\mathbf F([X,Y]^h,Z^h)-\mathbf F([Z,X]^h,Y^h)-\mathbf F([Y,Z]^h,X^h)\big)\vert_q,
\end{align*}
where $X^h,Y^h,Z^h$ are the horizontal lifts of the vector fields $X,Y,Z$ on $S$,
$q\in\pi^{-1}(x)$ is arbitrary.
Thus, $d\mathbf f=0$ if and only if \eqref{uslovDinv} is satisfied.
Consider the adapted coordinates $q=(q_1,\dots,q_{n+k})$ on $Q$ described in Subsection \ref{sec:gyroscopicfibred} and take
\[
X=\frac{\partial}{\partial q_i}, \qquad Y=\frac{\partial}{\partial q_j}, \qquad Z=\frac{\partial}{\partial q_p}, \qquad 1\le i,j,p \le n.
\]
Then the equation $d\mathbf f(X,Y,Z)=0$ takes the form \eqref{uslovD}.\end{proof}

\begin{rem}
In the special case, when $\mathbf F=\pi^* \mathbf w$, where $\mathbf w$ is a two-form on the base manifold $S$,
the equations \eqref{uslovD} are automatically satisfied ($F_{i,n+\nu}=0$, $1\le i \le n$, $1\le \nu \le k$). In this special case $\mathbf f=\mathbf w$,
and
$d\mathbf F=0$ if and only if $d\mathbf f=0$.
\end{rem}

\section{Almost Hamiltonian description and an invariant measure}\label{sec4}

\subsection{ Almost symplectic manifolds}
Recall that an \emph{almost symplectic structure} is a pair $(M,\omega)$ of a manifold $M$ and a nondegenerate 2-form $\omega$ (see \cite{LM}).
Here we do not assume that the form $\omega$ is closed, in contrast to the symplectic case.
As in the symplectic case, since $\omega$ is nondegenerate, to a given function $H$ one can associate the \emph{almost Hamiltonian vector field} $X_H$ by the identity
\[
i_{X_H}\omega(\cdot)=\omega(X_H,\cdot)=-dH(\cdot).
\]

The almost symplectic structure $(M,\omega)$ is \emph{locally conformally symplectic}, if in a neighborhood of each point $x$ on $M$, there exists a function $f$ different from zero such that
$f\omega$ is closed. If $f$ is defined globally, then  $(M,\omega)$ is \emph{conformally symplectic} \cite{LM}.

\subsection{Reduced flows on cotangent bundles}

Let $(x_1,\dots,x_n)$ be local coordinates on $S$ in which the metric $\mathbf g$ is given by the quadratic form $\sum_{ij} g_{ij} dx_i\otimes dx_j$
and the components of the (1,2)-tensor $\mathbf C$ are $C^k_{ij}$ (see Remark \ref{giroskopskiTenzor}).
Then the Lagrangian, the gyroscopic two-form and the {\bf JK}-term read as follows
\begin{align*}
&l(x,\dot x)= \frac{1}{2}\sum  g_{ij} \dot x_i\dot x_j-v(x), \\
&\mathbf f=\sum_{i<j} f_{ij} dx_i\wedge dx_j,\\
&\mathbf{JK}(X,Y)\vert_{(x,\dot x)}=\mathbf g(\dot x,\mathbf C(X,Y))=\sum_{k,l,i,j} g_{kl}C^k_{ij}X_iY_j\dot x_l.
\end{align*}

We also introduce the Hamiltonian function
\[
h(x,p)=\frac12 (p, \mathbf g^{-1}(p))+v(x)=\frac12 \sum g^{ij} p_ip_j+v(x),
\]
as the usual Legendre transformation of $l$. Here $(p_1,\dots,p_n,x_1,\dots,x_n)$ are
the canonical coordinates  of the cotangent bundle $T^*S$,
\[
p_i=\partial l/\partial \dot x_i=\sum_j g_{ij}\dot x_j,
\]
and $\{g^{ij}\}$ is the inverse of the metric matrix $\{g_{ij}\}$.
For simplicity,  the same symbol denotes a function on the base manifold $f\colon S\to \R$ and its lift to the cotangent bundle
$\rho^*f=f\circ \rho\colon T^*S\to \R$, where $\rho:T^*S\to S$ is the canonical projection.

In canonical coordinates the equations \eqref{ChaplyginRed} take the form
\begin{align}
\label{ham1}& \dot x_i=\frac{\partial h}{\partial p_i}=\sum_{j=1}^n g^{ij} p_j,\\
\label{ham2}& \dot p_i=-\frac{\partial h}{\partial x_i}+\Pi_i(x,p)+\sum_{j=1}^n f_{ij}(x) \frac{\partial h}{\partial p_j}.
\end{align}
Here, the {\bf JK}-term is given in the form
\begin{equation}\label{hamPi}
\Pi_i(x,p)=\mathbf{JK}\big(\frac{\partial}{\partial x_i},\dot x\big)\vert_{\dot x=\mathbf g^{-1}(p)}=\sum_{k,l,j=1}^n g_{kl}\dot x_l C^k_{ij}(x)\dot x_j\vert_{\dot x=\mathbf g^{-1}(p)}=\sum_{j,k=1}^n C^k_{ij}(x)p_k\frac{\partial h}{\partial p_j}.
\end{equation}

Let $z=(x,p)$. The reduced equations \eqref{ham1}, \eqref{ham2} on the cotangent bundle
$T^*S$ can be written in the almost Hamiltonian form
\begin{equation}\label{red:eq}
\dot z=X_{red}, \qquad i_{X_{red}}(\Omega+\sigma+\rho^* \mathbf  f)=-dh,
\end{equation}
where $\Omega$ is the canonical symplectic form on $T^*S$,
$\sigma$ is a semi-basic form defined by the {\bf JK} term (see \cite{CCL2002, St}):
\begin{align}\label{eq:Omega}
& \Omega=dp_1\wedge dx_1+\dots+dp_n\wedge dx_n, \\
& \sigma=\sum_{1\le i<j\le n} \sum_{k=1}^n C^k_{ij}(x) p_k dx_i\wedge dx_j.
\end{align}

\subsection{Invariant measure}
The existence of an invariant measure for nonholomic problems is well studied  (see
\cite{Fe1988, VeVe2, Kozlov,  FK1995, ZenkovBloch, FNM, Jov2015, FGM2015}).
We will consider smooth measures  of the form $\mu=\nu\,\Omega^n$, where $\Omega^n$ (see \eqref{eq:Omega}) is the standard measure on the cotangent bundle $T^*S$ and $\nu$ is a nonvanishing smooth function, called {\it the density} of the measure $\mu$.

In absence of potential and gyroscopic forces, it was proved in \cite{CCL2002} that the equations \eqref{ham1}, \eqref{ham2}
 have an invariant measure if and only if  its density is basic, i.e, $\nu=\nu(x)$.
 Then the system with a potential force $v(x)$ also preserves the same measure (see \cite{St, CCL2002})).

For $\mathbf f=0$, the existence of the basic density $\nu=\nu(x)$ is equivalent to the condition that the one-form
\begin{equation}\label{Theta}
\Theta=\sum_{i,j}C^j_{ij}(x)dx_i, \quad \text{i.e.,} \quad \Theta(X)\vert_x=\tr\mathbf C(X,\cdot)\vert_x, \quad X\in T_x S,
\end{equation}
is exact: there exists a function $\lambda$ such that $\Theta=d\lambda$. Then the function $\nu(x)=\exp(\lambda(x))$ is the density of an invariant measure (see \cite{CCL2002, Naranjo2020}).
The statement formulated in terms of the
tensor field $\mathbf B$ is given in \cite{CCL2002}, while in
\cite{Naranjo2020} it is formulated in terms of the gyroscopic tensor $\mathcal T=-\mathbf C$. An example of a system with a potential force and with an invariant non-basic measure is also given in
\cite{Naranjo2020}.

In the presence of the gyroscopic form $\mathbf f$ we have a similar situation.

\begin{thm}\label{novamera}
The reduced gyroscopic Chaplygin equations \eqref{ham1}, \eqref{ham2} have an invariant measure $\mu=\nu\,\Omega^n$ with a basic density $\nu(x)$
if and only if the one-form \eqref{Theta} is exact $\Theta=d\lambda$. Then the function $\nu=\exp(\lambda(x))$ is the density of the invariant measure.
\end{thm}

In other words, according to \cite{CCL2002, Naranjo2020}, a Chaplygin system with a gyroscopic term possesses
a basic invariant measure if and only if the same Chaplygin system without gyroscopic
term preserves the same basic invariant measure.

\begin{proof}
The Lie derivative $\mathcal L_{X_{red}}(\mu)$ vanishes if and only if the divergence of the vector field $\nu X_{red}$ with respect to the canonical measure equals to zero.
 By using the identities
\[
\frac{\partial}{\partial p_i}\frac{\partial h}{\partial p_j}=g^{ji}, \quad \sum_{i,j}f_{ij}g^{ji}=0, \quad \sum_{ij} C^k_{ij} g^{ji}=0,
\]
we get:
\begin{align*}
\mathrm{div}(\nu X_{red})=&\sum_{i=1}^n\frac{\partial}{\partial x_i}\big(\nu\frac{\partial h}{\partial p_i}\big)+\sum_i\frac{\partial}{\partial p_i}\big(\nu\big(-\frac{\partial h}{\partial x_i}+\sum_{k,j=1}^n C^k_{ij}(x) p_k\frac{\partial h}{\partial p_j}+\sum_{j=1}^n f_{ij}(x) \frac{\partial h}{\partial p_j}\big)\big)\\
=&\sum_{i=1}^n\big(\frac{\partial\nu}{\partial x_i}-\nu\sum_{j=1}^n C^j_{ij}(x)\big)\frac{\partial h}{\partial p_i}.
\end{align*}

Since $\dot x_i=\frac{\partial h}{\partial p_i}$ is arbitrary  for each fixed $x$,
 the vector filed $X_{red}$ preserves the measure $\nu\,\Omega^n$  if and only if
\[
\nu^{-1}\frac{\partial\nu}{\partial x_i}=\sum_{j=1}^n C^j_{ij}(x), \qquad i=1,\dots,n,
\]
that is, if and only if
\[
d\,{\ln\nu}=\sum_{i=1}^n \nu^{-1}\frac{\partial\nu}{\partial x_i} dx_i=\sum_{i,j=1}^n C^j_{ij}(x) dx_i=\Theta.
\]
Note that, although the proof is derived in local coordinates, all considered objects
are global and the identity $d\,{\ln\nu}=\Theta$ holds globally.\end{proof}

\section{Chaplygin Hamiltonization for systems with magnetic forces}\label{sec5}

\subsection{Chaplygin multipliers in the Lagrangian framework}\label{CMLF}

We consider the reduced Chaplygin systems   \eqref{ChaplyginRed}  and study the question of their transformation into a Lagrangian system after a time reparametrization.

Let us consider a time substitution $d\tau={\mathcal N}(x)dt$, where
${\mathcal N}(x)$ is a differentiable nonvanishing function on $S$.
Denote $x^{\prime}={dx}/{d\tau}=\N^{-1}\dot x$.

We first treat the exact case: $\mathbf f=d\mathbf a$ (see Remark \ref{primedba}). Locally, the one-form $\mathbf a$ is given by
$\mathbf a=\sum_i a_i(x)dx_i$ and
\begin{equation*}
l_1(x,\dot x)= \frac{1}{2}\sum  g_{ij} \dot x_i\dot x_j+\sum_i a_i\dot x_i-v(x).
\end{equation*}

The Lagrangians $l$ and $l_1$
in the coordinates
$(x,x^{\prime})$ are denoted $l^*$ and $l_1^*$ respectively and
take the form
\begin{align}\label{eq:l*}
&l^*(x,x^{\prime})= \frac{1}{2}\sum \mathcal{\ N}^2 g_{ij} x_i^{\prime}x_j^{\prime}-v(x), \\
&l_{1}^*(x,x^{\prime})= \frac{1}{2}\sum \mathcal{\ N}^2 g_{ij} x_i^{\prime}x_j^{\prime}+\sum_i \N a_i(x)x_i^{\prime}-v(x).
\end{align}

Following  Chaplygin \cite{Chaplygin1911}, we are looking for a nowhere vanishing function
$\N(x)$, called \emph{a Chaplygin reducing multiplier} such that the reduced Chaplygin system \eqref{ChaplyginRed2}
\begin{equation}\label{orEQ1}
\frac{d}{dt}\frac{\partial l_1}{\partial \dot x_i}=\frac{\partial l_1}{\partial x_i}+\sum_{k,l,j=1}^n C^k_{ij}(x) g_{kl} \dot x_l \dot x_j
\end{equation}
after a time reparametrization $d\tau={\mathcal N}(x)dt$ becomes the Lagrangian system
\begin{equation}\label{modEQ1}
\frac{d}{d\tau}\frac{\partial l^*_1}{\partial x'_i}=\frac{\partial l^*_1}{\partial x_i}, \quad i=1,\dots,n.
\end{equation}

Equivalently, we can use the Lagrangians $l$ and $l^*$. Let
\begin{align*}
&\mathbf f=d\mathbf a=\sum_{i<j} f_{ij} dx_i\wedge dx_j, \quad f_{ij}=\frac{\partial a_j}{\partial x_i}-\frac{\partial a_i}{\partial x_j},\\
&\mathbf f^*=d(\N\mathbf a)= \sum_{i<j} f^*_{ij} dx_i\wedge dx_j, \quad f^*_{ij}=\N f_{ij}+a_j\frac{\partial \N}{\partial x_i}-a_i\frac{\partial \N}{\partial x_j}.
\end{align*}

Then, we are looking for a nowhere vanishing function $\N(x)$, such that the reduced Chaplygin system
\begin{equation}\label{orEQ}
\frac{d}{dt}\frac{\partial l}{\partial \dot x_i}=\frac{\partial l}{\partial x_i}+\sum_{k,l,j=1}^n C^k_{ij}(x) g_{kl} \dot x_l \dot x_j+\sum_{j=1}^n f_{ij}(x) \dot x_j
\end{equation}
after a time reparametrization $d\tau={\mathcal N}(x)dt$ becomes the Lagrangian system with magnetic forces
\begin{equation}\label{modEQ}
\frac{d}{d\tau}\frac{\partial l^*}{\partial x'_i}=\frac{\partial l^*}{\partial x_i}+  \sum_{j=1}^n  f^*_{ij} x_j', \quad i=1,\dots,n.
\end{equation}

\begin{prop} \label{L1prop} Suppose that $\mathbf f$ is exact: $\mathbf f=d\mathbf a$.
The reduced equations of the Chaplygin system with a linear term in velocities \eqref{orEQ1}
after a time reparametrization $d\tau={\mathcal N}(x)dt$ becomes the Lagrangian system
\eqref{modEQ1} if and only if the
corresponding system without the linear term allows the Chaplygin multiplier $\N(x)$ and
$d\N\wedge \mathbf a=0$, that is, if
\begin{equation}\label{condA}
a_j\frac{\partial \N}{\partial x_i}=a_i\frac{\partial \N}{\partial x_j}.
\end{equation}
\end{prop}

Note that conditions \eqref{condA} imply that
\[
\mathbf f^*=d(\N\mathbf a)=\N d\mathbf a+d\N\wedge \mathbf a=\N\mathbf f
\]
 and
\begin{equation}\label{df*}
d(\N\mathbf f)=d\N\wedge \mathbf f=0.
\end{equation}

Let us now turn to the non-exact case. Thus, we assume now $\mathbf f$ is not exact. In this case we set
\begin{equation}\label{fzvezda}
\mathbf f^*=\N\mathbf f.
\end{equation}

\begin{prop} \label{LAGprop} Suppose that $\mathbf f$ is not exact.
The equations of motion of the reduced gyroscopic  Chaplygin system \eqref{orEQ} after a time reparametrization $d\tau=\N dt$ become the Lagrangian equations
with gyroscopic forces \eqref{modEQ}, where $f^*$ is given by \eqref{fzvezda}
if and only if the corresponding system without gyroscopic forces allows the Chaplygin multiplier $\N(x)$.
\end{prop}

Propositions \ref{L1prop} and \ref{LAGprop} follow from the derivation given below for the Hamiltonian setting as indicated in Remark \ref{dokazLAG}.

Note that the gyroscopic system \eqref{modEQ} is magnetic if the form \eqref{fzvezda} is closed. In particular,
if $\mathbf f$ is closed, but not exact, then the Lagrangian system \eqref{modEQ} is magnetic only if the condition \eqref{df*} holds.
The condition \eqref{df*} is always satisfied when $n=2$. This is a rather strong condition for $n\ge 3$. When $n=3$, condition \eqref{df*} reduces to the partial differential equation
\[
f_{23}\frac{\partial\N}{\partial x_1}+f_{31}\frac{\partial\N}{\partial x_2} + f_{12}\frac{\partial\N}{\partial x_3}=0.
\]

Finally, it is important to note that even if we consider the exact case $\mathbf f=d\mathbf a$ and the Lagrangians that are linear in velocities, instead of
the equations \eqref{orEQ1} and \eqref{modEQ1} and the gyroscopic form defined by $\mathbf f^*=d(\N\mathbf a)$ it is more natural to consider
the equations \eqref{orEQ} and \eqref{modEQ} with  $\mathbf f^*$ defined as $\mathbf f^*=\N\mathbf f=\N d\mathbf a$. In the latter case, for $n=2$, the form $\mathbf f^*$ is magnetic regardless of
\eqref{condA}.

\subsection{Confomally symplectic structures}
The existence of an invariant measure is closely related to the Hamiltonization problem for magnetic $G$--Chaplygin systems.
We first consider $G$-Chaplygin systems without the gyroscopic term, see \cite{CCL2002, St, EKMR2005}.
For $\mathbf f\equiv 0$, the reduced system \eqref{red:eq} takes the form
\begin{equation}\label{red:eq0}
\dot z=X^0_{red}, \qquad i_{X^0_{red}}(\Omega+\sigma)=-dh.
\end{equation}
 Suppose that the form $\Omega+\sigma$ is conformally symplectic, i.e. there exists a nonvanishing  function $\N$, such that
$d(\mathcal N(\Omega+\sigma))=0$. Since $d\Omega=0$, the last relation can be rewritten as:
\begin{equation}\label{faktor1}
d\N\wedge\Omega+d\N\wedge\sigma+\N d\sigma=0.
\end{equation}

After the time rescaling $d\tau=\N dt$, the equation \eqref{red:eq0} reads
\[
z'=\N^{-1}\dot z=\N^{-1} X^0_{red}=:\tilde X^0_{red}.
\]
The last relation introduces the rescaled vector field $\tilde X^0_{red}$, which is Hamiltonian:
\[
i_{\tilde X^0_{red}}\N(\Omega+\sigma)=-dh.
\]

Therefore, the system in the new time becomes the Hamiltonian system
with respect to the symplectic form $\N(\Omega+\sigma)$.
Then,  according to the Liouville theorem \cite{Ar},  the Hamiltonian vector field $\tilde X^0_{red}$ preserves the standard
measure $\N^{n}(\Omega+\sigma)^n=\N^n\Omega^n$,
\[
 \mathcal L_{\tilde X^0_{red}}(\N^n\Omega^n)=d(i_{\tilde X^0_{red}}(\N^n \Omega^n))=0.
\]
Thus,  for the almost Hamiltonian vector field $X^0_{red}=\N\tilde X^0_{red}$ we have
\[
\mathcal L_{X^0_{red}}(\N^{n-1}\Omega^n)=d(i_{X^0_{red}}(\N^{n-1} \Omega^n))=d(i_{\tilde X^0_{red}}(\N^n \Omega^n))=0,
\]
and the flow of $X^0_{red}$  preserves the measure $\N^{n-1}\Omega^n$.

Now, we consider $G$--Chaplygin systems with a gyroscopic term.

\begin{prop}\label{BazniMnozilac}
The function $\N=\N(x)$ is a conformal factor for the almost symplectic form $\Omega+\sigma+\rho^*\mathbf f$ if and only if it is a conformal factor for
the almost symplectic form $\Omega+\sigma$ and the form $\mathbf f^*=\N \mathbf f$ is magnetic.
\end{prop}

\begin{proof}
The form  $\Omega+\sigma+\rho^*\mathbf f$ is conformally symplectic
with a  conformal factor $\N$ if and only if
\begin{equation}\label{faktor2}
d\N\wedge\Omega+d\N\wedge \sigma+\N d\sigma+ d\N\wedge\rho^*\mathbf f+\N\rho^*d\mathbf f=0.
\end{equation}

Assume that $\N=\N(x)$ is basic.
Since only two last terms are basic, equation \eqref{faktor2} is satisfied if and only if $\N(x)$ satisfies
\eqref{faktor1} and  $\mathbf f^*=\N \mathbf f$ is closed.
\end{proof}

Consider the reduced gyroscopic Chaplygin system \eqref{red:eq}.
If $\N=\N(x)$ is a conformal factor for $\Omega+\sigma+\rho^*\mathbf f$, as above we have
that  the rescaled vector field $\tilde X_{red}=\N^{-1}X_{red}$ is Hamiltonian and preserves the
measure $\N^{n}(\Omega+\sigma+\rho^* \mathbf f)^n=\N^n\Omega^n$. Thus, the reduced gyroscopic Chaplygin system $\dot z=X_{red}$
preserves the same measure as in the case of the absence of gyroscopic forces. This is in accordance with Theorem \ref{novamera}.

The existence of a basic conformal factor, as we will see in Subsection \ref{LEG}, is equivalent to the condition that
$\N$ is the classical Chaplygin multiplier in the Lagrangian framework described above.

\subsection{Chaplygin multipliers: from the Lagrangian to the Hamiltonian framework}\label{LEG}

In the study of  nonholonomic rigid body systems in $\R^n$ (see \cite{FJ2004, Jo4, Jov2018b, Jov2019})
the Chaplygin time reparametrization of Lagrangian systems was transported into the Hamiltonian framework via the Legandre transformation.
Similarly, consider the time substitution $d\tau={\mathcal N} (x)dt$ and the Lagrangian function $l^*(x,x')$ given in \eqref{eq:l*}.
Then the conjugate momenta are
\[
\tilde p_i=\partial l^*/\partial x'_i=\N^2\sum_j g_{ij}x'_j,
\]
and the corresponding Hamiltonian is
\[
\quad h^*(x,\tilde p)=\frac12\sum \frac{1}
{\mathcal{N}^2} g^{ij} \tilde p_i \tilde p_j+v(x).
\]

The following diagram commutes:
\begin{equation}\label{dijagram}
\begin{CD}
 TS\{x,\dot x\} @ > x^{\prime}=\N^{-1}\dot x >> TS\{x,x^{\prime}\} \\
 @ V p=\mathbf g (\dot x) VV   @ VV \tilde p={\mathcal N}^2 \mathbf g (x^{\prime}) V \\
 T^*S\{x,p\}  @ >  \tilde p=\mathcal{\ N} p >>  T^*S\{x,\tilde p\} .
\end{CD}
\end{equation}

Let $\tilde\Omega$ be the canonical symplectic form on $T^*S$ with respect to the coordinates $(x,\tilde p)$. Then
\begin{equation}\label{symCL1}
\tilde\Omega=\sum_i d\tilde p_i\wedge dx_i=\N\Omega+d\N\wedge \theta, \quad \theta=p_1dx_1+\dots p_n dx_n, \quad \Omega=d\theta.
\end{equation}

Thus, $h$ and $h^*$ represent the same Hamiltonian function on $T^*S$ written in two coordinate systems. These coordinate systems are  related
by the non-canonical change of variables
\begin{equation}\label{nonCAN}
(x,p)\longmapsto (x,\tilde p)=(x,\N p).
\end{equation}

Assume that the two-form $\mathbf f^*=\N \mathbf f$ is closed on $S$.

By using the commutative diagram \eqref{dijagram}, we get that
the function $\N$ is a Chaplygin reducing multiplier for the reduced gyroscopic Chaplygin system \eqref{orEQ} (see Subsection \ref{CMLF}) if
and only if the almost Hamiltonian  equations \eqref{ham1}, \eqref{ham2}, after the time reparametrisation
$d\tau={\mathcal N}(x)dt$ and the coordinate transformation  \eqref{nonCAN}
become the Hamiltonian equations
\begin{equation}\label{newEq}
x'_i=\frac{\partial h^*}{\partial \tilde p_i}(x,\tilde p),\qquad
 \tilde p'_i=-\frac{\partial h^*}{\partial x_i}(x,\tilde p)+\N\sum_j f_{ij}(x)\frac{\partial h^*}{\partial \tilde p_j}(x,\tilde p)
\end{equation}
with respect to the twisted symplectic form
\begin{equation}\label{symCL}
\tilde\Omega+\rho^* \mathbf  f^*=\sum_i d\tilde p_i\wedge dx_i+\N\sum_{i<j}f _{ij} dx_i\wedge dx_j.
\end{equation}

Let $\N$ be a nonvanishing function and consider the time reparametrisation
$d\tau={\mathcal N} (x)dt$.  The equations \eqref{newEq} in the original time $t$, after the coordinate transformation  \eqref{nonCAN} take the form
\begin{align}
\label{x11} \dot x_i &= \N \frac{\partial h^*}{\partial \tilde p_i}(x,\tilde p)=\N\N^{-2}\sum_j g^{ij}\tilde p_j=\sum_j g^{ij}p_j ,\\
\label{x10} \dot{\tilde p}_i &= - \N\frac{\partial h^*}{\partial x_i}(x,\tilde p)+\N^2\sum_j f_{ij}(x)\frac{\partial h^*}{\partial \tilde p_j}(x,\tilde p)\\
\nonumber &=-\N\big(\frac{1}{2\N^2}\sum_{j,k}\frac{\partial g^{jk}}{\partial x_i}\tilde p_j\tilde p_k-\frac{1}{\N^{3}}\frac{\partial\N}{\partial x_i}\sum_{j,k} g^{jk} \tilde p_j\tilde p_k+\frac{\partial v}{\partial x_i}-\sum_{j,k} f_{ij} g^{jk}p_k\big)\\
\nonumber &=-\N\big(   \frac{1}{2}\sum_{j,k}\frac{\partial g^{jk}}{\partial x_i}p_j p_k-\frac{1}{\N}\frac{\partial\N}{\partial x_i}\sum_{j,k} g^{jk} p_j p_k +\frac{\partial v}{\partial x_i}    -\sum_{j,k} f_{ij} g^{jk}p_k     \big)
\end{align}

The equations \eqref{ham1} and \eqref{x11} coincides.
From $\tilde p_i=\N p_i$, we get
$\dot{\tilde p}_i=\N\dot p_i+\dot\N p_i$, that is $\dot p_i=\N^{-1}(\dot{\tilde p}_i-\dot\N p_i)$.
Therefore, using equation \eqref{x10}, we obtain
\begin{align}
 \dot p_i & =-\frac{1}{2}\sum_{j,k}\frac{\partial g^{jk}}{\partial x_i}p_j p_k+\frac{1}{\N}\frac{\partial\N}{\partial x_i}\sum_{j,k} g^{jk} p_j p_k
              -\frac{1}{\N}\sum_j \frac{\partial\N}{\partial x_j}\dot x_j p_i-\frac{\partial v}{\partial x_i}+\sum_{j,k} f_{ij} g^{jk}p_k \label{HAMver}  \\
         & =  -\frac{\partial h}{\partial x_i}(x,p)+\frac{1}{\N}\frac{\partial\N}{\partial x_i}\sum_{j,k} g^{jk} p_j p_k
              -\frac{1}{\N}\sum_{j,k}\frac{\partial\N}{\partial x_j}g^{jk}p_k p_i +\sum_{j,k} f_{ij} g^{jk}p_k \nonumber  \\
 &= -\frac{\partial h}{\partial x_i}(x,\tilde p)+\sum_{j,k,l=1}^n \N^{-1}\big( \delta^k_j \frac{\partial\N}{\partial x_i}-\delta^k_i \frac{\partial\N}{\partial x_j}   \big)g^{jl} p_kp_l+\sum_{j=1}^n f_{ij}(x) \frac{\partial h}{\partial p_j}, \nonumber
 \end{align}

The equations \eqref{ham2}, \eqref{hamPi}, and \eqref{HAMver} imply that
the reduced gyroscopic Chaplygin system \eqref{ham1}, \eqref{ham2} after the time reparametrization $d\tau=\N(x)dt$ and the change
of variables \eqref{nonCAN} takes the twisted canonical form \eqref{newEq} if and only if we have the equality of the quadratic forms in momenta:
\begin{equation}\label{simple*}
\sum_{j,k,l=1}^n C^k_{ij}(x)g^{jl}p_kp_l=\sum_{j,k,l=1}^n \N^{-1}\big( \delta^k_j \frac{\partial\N}{\partial x_i}-\delta^k_i \frac{\partial\N}{\partial x_j}   \big)g^{jl} p_kp_l,\quad i=1,\dots,n.
\end{equation}
In the invariant form, \eqref{simple*} can be written as the condition on {\bf JK} force term \eqref{hamPi}:
\begin{equation}
\Pi(x,p)=\N^{-1}(p,\mathbf g^{-1}(p))d\N-\N^{-1}(d\N,\mathbf g^{-1}(p)) p,\label{x5}
\end{equation}

\begin{rem}\label{dokazLAG}
Note that the equations \eqref{newEq}--\eqref{HAMver} are valid without assumption that the form $\mathbf f^*=\mathcal N\mathbf f$ is closed, i.e., when
$\tilde\Omega+\rho^* \mathbf  f^*$ (see \eqref{symCL}) is an almost symplectic form as well. In this way, according to the commutative diagram \eqref{dijagram},
they imply Propositions \ref{L1prop} and \ref{LAGprop}.
\end{rem}

It is clear that the sufficient conditions for the identities \eqref{simple*} are:
\begin{equation}\label{simple1}
C^k_{ij}(x)=\N^{-1}\big( \delta^k_j \frac{\partial\N}{\partial x_i}-\delta^k_i \frac{\partial\N}{\partial x_j}   \big), \qquad i,j,k=1,\dots,n.
\end{equation}

Thus,
if the (1,2)-tensor field $\mathbf C$ defined in Remark \ref{giroskopskiTenzor} satisfies \eqref{simple1}, $\N$ is a Chaplygin reducing multiplier
for the reduced gyroscopic $G$--Chaplygin system \eqref{ham1}, \eqref{ham2}, e.g., \eqref{orEQ}.
Then the $(1,2)$-tensor $\mathbf C$ and the two-form $\sigma$ in the invariant form can be written as
\begin{align}
&\mathbf C(X,Y)=\N^{-1}X(\N)Y-\N^{-1}Y(\N)X,\label{simple}\\
&\sigma=\N^{-1} d\N \wedge \theta.\label{sigmaN}
\end{align}

Moreover, from \eqref{symCL1}, \eqref{symCL},  and \eqref{sigmaN}, we obtain that  the form $\Omega+\sigma+\rho^* \mathbf  f$ is conformally symplectic with $\mathcal N$ a conformal factor being a Chaplygin reducing multiplier:
\[
\tilde\Omega+\rho^*\mathbf f^*=\N(\Omega+\sigma+\rho^* \mathbf  f).
\]

In the terminology of \cite{Naranjo2019a, Naranjo2019b}, the equations \eqref{simple1} and \eqref{simple} mean that
the gyroscopic tensor $\mathcal T=-\mathbf C$ is $\phi$--simple, where $\phi=\ln\N$. Following Garcia-Naranjo, we say that a (1,2)-tensor $\mathbf C$ is
$\ln\N$-\emph{simple} if \eqref{simple} holds.

In \cite{Naranjo2020} the  following inverse   statement is proved: if a two-form $\Omega+\sigma$ is conformally symplectic with a basic
conformal factor $\N(x)$, then the gyroscopic tensor $\mathcal T$ is $\ln\N$-simple.
Now, based on the above considerations, we can reformulate  and extend  Theorem 3.21 from \cite{Naranjo2020}  on $\phi$-simple Chaplygin systems
as follows:

\begin{thm}\label{opsta}
{\rm (i)} Assume that two-form $\mathbf f^*=\N \mathbf f$ is closed on $S$.  The conditions (a), (b), and (c) listed below are equivalent. The conditions (d) and (e) are equivalent, while (e) implies (c):
\begin{itemize}
\item[(a)] the reduced gyroscopic Chaplygin system \eqref{orEQ} after the time reparametrization $d\tau=\N(x)dt$ takes the form of the magnetic Lagrangian system  \eqref{modEQ};
\item[(b)]  the reduced gyroscopic Chaplygin system \eqref{ham1}, \eqref{ham2} after the time reparametrization $d\tau=\N(x)dt$ and the change
of variables \eqref{nonCAN} takes the twisted canonical form \eqref{newEq};
\item[(c)]  the {\bf JK} force term \eqref{hamPi} on $T^*S$ has the form \eqref{x5};
\item[(d)] the almost symplectic form $\Omega+\sigma+\rho^* \mathbf  f$ is conformally
symplectic with the base conformal factor $\N(x)$  and $\sigma$ is given by \eqref{sigmaN};
\item[(e)] the (1,2)-tensor $\mathbf C$ is $\ln\N$-simple,  that is, it is given by \eqref{simple}.
\end{itemize}

{\rm (ii)} If $\N(x)$ is a Chaplygin multiplier, then the reduced equations of motion \eqref{ham1}, \eqref{ham2} possess the
base invariant measure
\begin{equation}\label{invMES}
\N^{n-1}\Omega^n.
\end{equation}

{\rm (iii)} If $n = 2$, then the statement {\rm (ii)} can be inverted: if the reduced equations of motion \eqref{ham1}, \eqref{ham2} possess the base
invariant measure
\[
\N(x)dp_1\wedge dp_2 \wedge dx_1\wedge dx_2,
\]
then, after the time reparametrization $d\tau=\N(x)dt$ the reduced equations become the usual
Hamiltonian equations on $T^*S$ with respect to the twisted symplectic form \eqref{symCL}.
For $n=2$, all items (a), (b), (c), (d), (e) are equivalent.
\end{thm}

Theorem \ref{opsta} relates the classical Chaplygin Hamiltonization (items (a), (b), (c) see \cite{Chaplygin1911, FJ2004}) and the
Chaplygin Hamiltonization within the framework of almost symplectic forms and the gyroscopic tensor field $\mathbf C$ (items (d) and (e), see \cite{CCL2002, Naranjo2020}).

For the Veselova problem on $SO(n)$ (see \cite{FJ2004}) it is proved in \cite{Naranjo2020} that (c) implies (d) as well. A similar statement can be proved for the nonholonomic problem of a ball rolling over a sphere considered in \cite{Jov2018b}.

\begin{rem}\label{rem:ref}
Note that \eqref{simple*} implies that the symmetric parts of the tensors
\[
\sum_{j=1}^n C^k_{ij}(x)g^{jl} \quad \text{and} \quad  
\sum_{j=1}^n \N^{-1}\big( \delta^k_j \frac{\partial\N}{\partial x_i}-\delta^k_i \frac{\partial\N}{\partial x_j}   \big)g^{jl}
\]
are equal, but the conditions \eqref{simple*} and \eqref{simple1}, i.e, the items (c) and (e) of Theorem \ref{opsta} do not need to be equivalent.
For example, one can have $\mathbf C$ and $\sigma$ different from zero, but with $i_{X_{red}}\sigma = 0$.
Then $\Pi= 0$ and $X_{red}$ is a Hamiltonian vector field with respect to the magnetic symplectic form $\Omega +\rho^*\mathbf f$. Thus, the constant
$\N= 1$ can be chosen as a Chaplygin multiplier.
As a result, the right hand side of \eqref{simple1} is zero, while the left hand side of \eqref{simple1} is different from zero.
\end{rem}

Further, from  Theorem \ref{opsta} it follows that if a Chaplygin system without gyroscopic force allows Hamiltonization with
a basic multiplier $\mathcal N$, and if $\mathcal N\mathbf f$ is closed, then the system with reduced gyroscopic force $\mathbf f$ also
allows Hamiltonization and vice versa: if a Chaplygin system with gyroscopic force $\mathbf f$ allows Hamiltonization with
a basic multiplier $\mathcal N$ (either in the sense that $\mathcal N$ is a conformal factor for the almost symplectic form $\Omega+\sigma+\rho^*\mathbf f$ and according  Proposition \ref{BazniMnozilac} $\mathcal N\mathbf f$ is closed,  or
in the sense of the classical Hamiltonization where $\mathcal N\mathbf f$ is also closed) then the system without the gyroscopic force $\mathbf f$ allows Hamiltonization as well.

For $n=2$, the equations \eqref{orEQ} are
\begin{align}
&\frac{d}{dt}\frac{\partial l}{\partial \dot x_1}=\frac{\partial l}{\partial x_1}+S(x)\dot x_2, \label{ln1}\\
&\frac{d}{dt}\frac{\partial l}{\partial \dot x_2}=\frac{\partial l}{\partial x_2}-S(x)\dot x_1, \label{ln2}
\end{align}
where
\[
S(x)=\sum_{k,l=1}^2 C^k_{12}(x) g_{kl}\dot x_l+ f_{12}(x).
\]

Item (iii) of Theorem \ref{opsta} is given in \cite{BoMa, BBM2013},
where the Lagrangian systems of the form \eqref{ln1}, \eqref{ln2}, for $f_{12}(x) \neq 0$
are called \emph{generalised Chaplygin systems}.

\section{Chaplygin ball with a gyroscope rolling over a plane and over a sphere}\label{sec6}

\subsection{Chaplygin ball with a gyroscope rolling without slipping}\label{sec6.1}
One of the most famous solvable problems in nonholonomic mechanics describes
rolling without slipping of a balanced,  dynamically nonsymmetric   ball over a horizontal
plane (Chaplygin \cite{Ch1}). After \cite{Ch1}, a balanced,  dynamically nonsymmetric ball is called the \emph{Chaplygin ball}, see \cite{Kozlov2, AKN, BoMa, BM2008, BM, BN, BMT, BBM2013}.

Let $O_{\bf B}$, $a$, $m$, $\mathbb I=\diag(A,B,C)$, be the center, radius, mass and the inertia operator of a  ball $\bf B$.
There are three possible configurations in the problem of rolling without slipping of the Chaplygin ball $\bf B$ over a fixed sphere $\mathbf S$ of the radius $b$:

\begin{itemize}

\item[(i)] rolling of $\bf B$ over the outer surface of $\mathbf S$ and $\mathbf S$ is outside $\bf B$ (see the leftmost part of Fig. \ref{fig:3scen});

\item[(ii)] rolling of $\bf B$ over the inner surface of $\mathbf S$ ($b>a$)(see the central part of Fig. \ref{fig:3scen});

\item[(iii)] rolling of $\bf B$ over the outer surface of $\mathbf S$ and $\mathbf S$ is within $\bf B$; in this case
$b<a$ and the rolling ball $\bf B$ is a spherical shell (see the rightmost part of Fig. \ref{fig:3scen}).
\end{itemize}

Let $\varepsilon={b}/({b\pm a})$,
where we take "$+$" for the case (i) and "$-$" in the cases
(ii) and (iii) and let $D=m a^2$.
The equations of motion in the frame attached to the ball can be written in the form
\begin{equation}
\dot{\vec{\mathbf k}}=\vec{\mathbf  k}\times\vec\omega,
\qquad \dot{\vec{\gamma}}=\varepsilon
\vec\gamma\times\vec\omega,
\label{Chap}
\end{equation}
where $\omega$ is the angular velocity of the ball,
$\vec{\mathbf  k}=\mathbb I \vec\omega+ D\vec\omega-D \langle \vec\omega,\vec\gamma\rangle \vec\gamma$
is the angular momentum of the ball with respect to the point of contact, and $\gamma$ is the unit normal to the sphere $\mathbf S$ at the contact point.

When $b$ tends to infinity, then $\varepsilon$ tends to $1$ and $\vec\gamma$ tends to the unit vector that is constant in the fixed reference frame.
This way,  for $\varepsilon=1$,    we obtain the equations of motion of the Chaplygin ball rolling over the plane orthogonal to $\vec\gamma$.

\begin{figure}[ht]\label{fig:3scen}
\includegraphics[width=95mm]{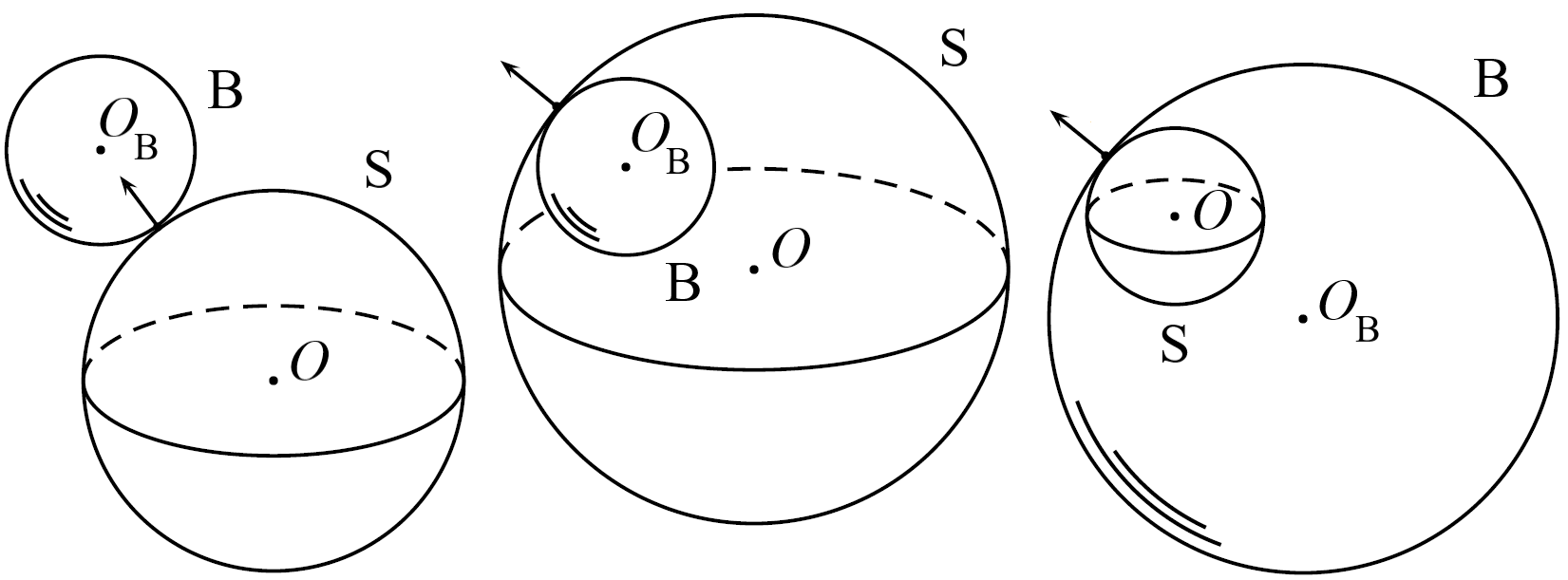}
\caption{Rolling of the ball $\bf B$ with center $O_{\bf B}$ over the sphere $\mathbf S$ with center $O$: three scenarios}
\end{figure}

An invariant measure of the system was derived by Chaplygin for $\varepsilon=1$ \cite{Ch1}, and by Yaroshchuk for $\varepsilon \ne 1$ \cite{Ya}. Remarkably, for $\varepsilon=-1$, which is the case (iii) above with $a=2b$,  the problem is integrable  (see \cite{BF, BFM, BM3}).

Next, we assume that a gyroscope is placed in a ball $\bf B$ such
that the mass center of the system coincides with the geometric center $O_{\bf B}$ of the ball.
The addition of a gyroscope  to the problem is equivalent to the addition of a constant angular momentum $\vec{\kappa}$
directed along the axis of the gyroscope to $\vec{\mathbf k}$ \cite{Bob1892,Zhuk1893}:
\begin{equation}
\frac{d}{dt}\big(\vec{\mathbf k}+\vec\kappa\big)=(\vec{\mathbf k}+\vec\kappa)\times\vec\omega,
\qquad \dot{\vec{\gamma}}=\varepsilon
\vec\gamma\times\vec\omega.
\label{Chap1}
\end{equation}
As above,
$
\vec{\mathbf  k}=\mathbb I \vec\omega+ D\vec\omega-D \langle \vec\omega,\vec\gamma\rangle \gamma,
$
where $D=a^2m$, $m$ is the mass of the system (ball with gyroscope), $\mathbb I$ is a new inertia operator that is described below (see \eqref{bobiljev}) together with the momentum $\kappa$ for the Bobilev symmetric case.

Markeev proved that the equations of motion for the rolling over the plane ($\varepsilon=1$) can be resolved in quadratures \cite{Markeev1985}.
The analysis of the bifurcation diagram and  the topology of the phase space of  the Markeev case is studied in \cite{Moskvin2009} and \cite{Zhila2020}, respectively.

There are two famous  classical cases of the system  \eqref{Chap1} for $\varepsilon=1$  where the quadratures are given in elliptic functions. These cases were studied by Bobilev \cite{Bob1892} and Zhukovskiy \cite{Zhuk1893}.

 In the Bobilev case the central ellipsoid of inertia of the ball $\bf B$ is rotationally symmetric
and the gyroscope axis coincides to the axis of symmetry.
Let $O_{\bf B}\vec{\mathbf e}_1\vec{\mathbf e}_2\vec{\mathbf e}_3$ and $O_{\bf B}\vec{\mathbf e}'_1\vec{\mathbf e}'_2\vec{\mathbf e}'_3$ be the moving frames attached to the ball $\bf B$
and the gyroscope in which
the inertia operator has the forms $\mathbb I_1=(A_1,A_1,C_1)$ and  $\mathbb I_2=(A_2,A_2,C_2)$, respectively.
It is assumed that  the axis of the gyroscope is fixed with respect to the ball
and coincides with the axis of symmetry of the inertia ellipsoid of the ball ($\vec{\mathbf e}_3=\vec{\mathbf e}'_3$) and
 that the forces applied to the gyroscope do not induce torque about the axis of the gyroscope.
Thus, the gyroscope  rotates with a constant angular velocity ${\omega'_3}$ about the axis of symmetry.
Then the operator $\mathbb I$  and the momentum $\vec\kappa$ in \eqref{Chap1} for the Bobilev case are given by:
\begin{equation}\label{bobiljev}
\mathbb I=\diag(A,A,C)=\diag(A_1+A_2,A_1+A_2,C_1) \quad \text{and} \quad \vec\kappa=C_2\omega'_3\vec{\mathbf e}_3.
\end{equation}

In the Zhukovskiy case there is an additional assumption, (\emph{called the Zhukovskiy condition}):
\begin{equation}\label{uslovZh}
C_1=A_1+A_2,
\end{equation}
that is,  it is assumed that  $\mathbb I$ is proportional to the identity matrix $\mathbb E=\diag(1,1,1)$.

Demchenko used the Zhukovskiy condition to integrate the problem of rolling of the gyroscopic ball over a sphere \cite{Dem1924} (see also \cite{DGJ}).
The integrability  of the problem of rolling of the gyroscopic ball over a sphere with the Bobilev conditions \eqref{bobiljev}  can be found  in  \cite{BoMa}.
The question about the existence of an integrable case for a dynamically nonsymmetric ball with a gyroscope rolling over a sphere is still  open.
 Another natural extension of the problem of the ball rolling over a sphere is recently given in \cite{DGJRCD, DGJ2022}.

\subsection{Chaplygin ball with a gyroscope rolling without slipping and twisting}\label{sec6.4}
One can consider the additional nonholonomic constraint $\langle\omega,\gamma\rangle=0$ describing \emph{no-twisting} condition: the ball $\bf B$ does not
rotate around the normal at the contact point and is called a \emph{rubber Chaplygin ball}. Then the momentum with respect to the contact point can be expressed
as  $\vec{\mathbf k}=\mathbf  I\vec\omega$, $\mathbf  I=\mathbb  I+D\mathbb  E$.
The gyroscopic equations take the form
\begin{equation}\label{e3a}
\frac{d}{dt}\big(\vec{\mathbf k}+\vec\kappa\big)=(\vec{\mathbf k}+\vec\kappa)\times \vec\omega+\lambda \vec\gamma, \qquad
\dot{\vec{\gamma}}=\varepsilon \vec\gamma\times\vec\omega,
\end{equation}
where the Lagrange multiplier is given by
$
\lambda=-\langle \vec\gamma,\mathbf  I^{-1}((\vec{\mathbf k}+\vec{\kappa})\times \vec\omega)\rangle /\langle \vec\gamma,\mathbf
I^{-1}\vec\gamma\rangle.
$

The system has an invariant measure with the same density as in the absence of a gyroscope
(see \cite{EKMR2005} for $\varepsilon=1$ and \cite{EK2007} for $\varepsilon \ne 1$). As in the Markeev integrable case,
for $\varepsilon=1$ the system is integrable according to the Euler-Jacobi theorem. This is proved by Borisov, Mamaev and Kilin in \cite{BoMa}
for the Veselova problem with a gyroscope, which is described by the same system of equations.
Borisov,  Bizyaev,  and Mamaev also pointed out  the integrability of the equations \eqref{e3a} for
$\varepsilon\ne 1$  in the case of the dynamical symmetry $A=B$ if the gyroscope
is oriented in the direction of the axis of the dynamical symmetry,  which gives the Bobilev conditions \eqref{bobiljev}
(see Table 2 in  \cite{BorBizMam2013}). Borisov and Mamaev  proved the
integrability of the problem without the gyroscope,  for $\varepsilon=-1$ \cite{BM2007}, providing analogy with the
non-rubber rolling.

The system of a Chaplygin ball with a gyroscope rolling without slipping and twisting over a sphere deserves to be studied in more detail. In order to describe its reduction and Hamiltonization, we will
consider a general problem in $\R^n$.

\section{The rolling of a gyroscopic ball without slipping and twisting in $\R^n$}\label{7}

\subsection{Rolling of a ball without slipping and twisting over a sphere}
The aim of this Section is to generalize the considerations from Section \ref{sec6} from $\R^3$ to $\R^n$, for any $n > 3$.
We start with the situation without gyroscopic or magnetic forces,
following \cite{Jov2018b, GajJov2019, GajJov2019b}. We consider in this Subsection the rolling without
slipping and twisting of an $n$-dimensional ball $\mathbf B$ of radius $a$ over the $(n-1)$-dimensional
fixed sphere $\mathbf S$ of radius $b$. There are three possible scenarios, in a full analogy with the three configurations described at the beginning of Section \ref{sec6.1} for $n=3$, recall Fig. \ref{fig:3scen}.

Consider the space frame $\R^n(\mathbf x)$ with the origin $O$ at the center of the fixed sphere $\mathbf S$ and the moving frame $\R^n(\mathbf X)$
with the origin $O_{\bf B}$ at the center of the rolling ball $\mathbf B$. The mapping from the moving to the space frame is given by $\mathbf x=g\mathbf X+\mathbf r$, where $g\in SO(n)$ is a rotation matrix and $\mathbf r=\overrightarrow{OO_{\bf B}}$
is the position vector of the ball center $O_{\bf B}$ in the space frame.
The configuration space $Q$ is the {direct product} of the Lie
group $SO(n)$ and the sphere $\mathcal S=\{\mathbf r\in\R^n\,\vert\,(\mathbf r,\mathbf r)=(b\pm a)^2\}$.

\begin{rem}
Here and below, we take the sign  "$+$"  for the case (i) and the case "$-$"  for the cases (ii) and (iii) of the three possible scenarios in analogy
with the three cases from the beginning of Section 6.1.
\end{rem}

Let $\omega=g^{-1}\dot g$ be the angular velocity of the ball in the moving frame, $m$ be the mass of the ball, and $\mathbb I: so(n) \to so(n)$ the inertia operator.
We additionally assume that the ball is balanced, i.e., its geometric center coincides with the mass center. We will call such a system a Chaplygin ball in $\R^n$.
Then the Lagrangian of the system is given by
\begin{equation}
L(g,\mathbf r,\omega,\dot{\mathbf r})=\frac12 \langle\mathbb I\omega,\omega\rangle+\frac12 m\langle \dot{\mathbf r},\dot{\mathbf r}\rangle,
\label{lagr}
\end{equation}
where now $\langle\cdot,\cdot\rangle$ is
the invariant scalar product proportional to the Killing form on $so(n)$ ($\langle\cdot,\cdot\rangle=-\frac12\tr(\cdot\circ\cdot)$)
and the Euclidean scalar product in $\R^n$, respectively.

The direction $\overrightarrow{OA}/\vert \overrightarrow{OA}\vert$ of the contact point $A$ in the frame attached to the ball is given by the unit vector
$\gamma=\frac{1}{b\pm a}g^{-1}\mathbf r$. It is invariant with respect to the diagonal
\emph{left} $SO(n)$-action:
$\tilde g\cdot (g,\mathbf r)=(\tilde g g,\tilde g\mathbf r)$, $\tilde g\in SO(n)$.  The action  defines $SO(n)$-bundle
\begin{equation}
\label{principal}
\xymatrix@R28pt@C28pt{
SO(n) \ar@{^{}->}[r]  & Q=SO(n) \times \mathcal S \ar@{^{}->}[d]^{\pi}  \\
  & S^{n-1}=Q/SO(n)  }
\end{equation}
with the submersion $\pi$ given by 
\[
\gamma=\pi(g,\mathbf r)=\frac{1}{b\pm a}g^{-1}\mathbf r.
\]

The contact point $A$ of the ball in the moving frame is $\mathbf X_A=-(\pm a\gamma)$.
The condition that the ball is rolling without
slipping is that the velocity $\dot{\mathbf x}_A$ of the contact point in the space frame is equal to zero
\[
0=\dot{\mathbf x}_A=\frac{d}{dt}\big(g\mathbf X_A+\mathbf r\big)=\mp a \dot g\gamma+\dot{\mathbf r}=\mp a(\dot gg^{-1})g\gamma+\dot{\mathbf r}.
\]

This leads to the constraint $\dot{\mathbf r}=\pm\frac{a}{b\pm a}\Omega \mathbf r$, where $\Omega=\Ad_g\omega=\dot gg^{-1}$ is the angular velocity in the space frame.
On the other hand, the condition of no twisting at the contact point can be written as the condition on $\Omega$: $\Omega\in \mathbf r\wedge\R^n$. The same condition can be written in terms of $\omega$: $\omega\in \gamma\wedge \R^n$. For more detail, see \cite{Jov2018b}.
The constraints determine the distribution
\[
\mathcal D_{(g,\mathbf r)}=\{(\omega,\dot{\mathbf r})\in T_{(g,\mathbf r)} SO(n)\times \mathcal S\,\vert\,
\dot{\mathbf r}=\pm\frac{a}{b\pm a}(\Ad_g\omega)\mathbf r, \,\omega\in g^{-1}\mathbf r\wedge\R^{n}\}
\]
of rank $(n-1)$, a principal connection of the bundle \eqref{principal}.
The Lagrangian $L$ from \eqref{lagr} is $SO(n)$--invariant as well. Thus,
an $n$--dimensional Chaplygin ball rolling without slipping and twisting
over a fixed sphere in $\R^{n}$ is a $SO(n)$--Chaplygin system. It
reduces to the tangent bundle $TS^{n-1}\cong \mathcal D/SO(n)$.

As in the three-dimensional case, we set $\varepsilon={b}/({b\pm a})$.
The horizontal lift $\dot\gamma^h\vert_{(g,\mathbf
r)}=(\omega,\mathbf V)$ is given by:
\begin{eqnarray*}
&& \omega=\frac{1}{\varepsilon}\gamma\wedge\dot\gamma,\\
&& \mathbf V=\dot{\mathbf r}=(b\pm
a)\frac{d}{dt}(g\gamma)=(b\pm a)\big(1-\frac{1}\varepsilon\big)g\dot\gamma.
\end{eqnarray*}

The  reduced Lagrangian $l$ and the $(0,3)$-tensor field $\Sigma$   are
(see \cite{Jov2018b})
\begin{align}
\label{eq:l}& l(\gamma,\dot\gamma) =\frac12 \mathbf g(\dot\gamma,\dot\gamma)=-\frac{1}{4\varepsilon^2}\tr(\mathbf
I(\gamma\wedge\dot\gamma)\circ(\gamma\wedge\dot\gamma) )=-\frac{1}{2\varepsilon^2}\langle\mathbf
I(\gamma\wedge\dot\gamma) \gamma,\dot\gamma\rangle,\\
\label{7.4}&\Sigma(X,Y,Z)\vert_\gamma=\frac{2\varepsilon-1}{2\varepsilon^3}\tr(\mathbf
I(\gamma\wedge X)\circ (Y\wedge Z))=\frac{2\varepsilon-1}{\varepsilon^3}\langle\mathbf
I(\gamma\wedge X)Y,Z\rangle,
\end{align}
 where, as in the three-dimension,
$\mathbf I=\mathbb I+D\cdot \mathrm{Id}_{so(n)}$ and $D=m a^2$.
We have
\begin{equation}\label{7.5a}
\frac{\partial l}{\partial
\gamma}=\frac{1}{\varepsilon^2}\mathbf
I(\gamma\wedge\dot\gamma)\dot\gamma, \quad \frac{\partial
l}{\partial \dot \gamma}=-\frac{1}{\varepsilon^2}\mathbf
I(\gamma\wedge\dot\gamma)\gamma, \quad \mathbf{JK}(\dot\gamma,\delta\gamma)=\frac{2\varepsilon-1}{\varepsilon^3}\langle\mathbf
I(\gamma\wedge \dot\gamma)\dot\gamma,\delta\gamma\rangle
\end{equation}

Therefore, the reduced Chaplygin equations \eqref{ChaplyginRed} without gyroscopic forces are:
\begin{equation}\label{REDUCED}
\delta l-\mathbf{JK}(\dot\gamma,\delta\gamma)=\Big\langle\frac{1}{\varepsilon^2}\frac{d}{dt}\big(\mathbf
I(\gamma\wedge\dot\gamma)\gamma\big)+\frac{1-\varepsilon}{\varepsilon^3}\mathbf
I(\gamma\wedge\dot\gamma)\dot\gamma,\delta\gamma\Big\rangle =0 , \qquad \delta\gamma\in
T_\gamma S^{n-1}.
\end{equation}

\begin{rem}\label{nemaReakcije}
Note that if the radii of the sphere and the ball are equal, then $\varepsilon=1/2$. Then, the curvature of
$\mathcal D$ vanishes and $\Sigma\equiv 0$ \cite{Jov2018b}.
For $n=3$, see \cite{EK2007, BMT}.
Also, if $\mathbb I$ is proportional to the identity operator then $\Sigma\equiv 0$. Then the   {\bf JK}-term   vanishes although the curvature of $\mathcal D$ is different from zero. Under these conditions, the reduced system is Hamiltonian without any time reparametrization.
\end{rem}

\subsection{Gyroscopic ball} Now, we want to consider the gyroscopic Chaplygin ball in $\R^n$  and to study how the addition of a gyroscopic term is going to
modify the reduced  equations of motion \eqref{REDUCED}. The equations \eqref{e3a} without the gyroscope have an analog in in $\R^n$:
\begin{equation}\label{ORIGINAL}
\dot{\mathbf k}=[\mathbf k,\omega]+\lambda_0, \qquad \dot\gamma=-\varepsilon\omega\gamma.
\end{equation}
Here $\mathbf k=\mathbf I\omega$ and the Lagrange multiplier $\lambda_0\in(\R^n\wedge\gamma)^\perp$ is determined from the condition that
$\omega\in\R^n\wedge\gamma$ (see \cite{Jov2018b}).

Let us notice that  the equations \eqref{REDUCED} alternatively can be derived directly by the substitution of $\omega=\frac{1}{\varepsilon}\gamma\wedge\dot\gamma$
in the equations \eqref{ORIGINAL}. The equations \eqref{ORIGINAL} are also a convenient starting point for  gyroscopic generalizations. With a suitable modification of $\mathbb I$ for the gyroscopic ball, the analogue of the equation \eqref{e3a} in $\R^n$ is
\begin{equation}\label{ORIGINALgyr}
\dot{\mathbf k}=[\mathbf k,\omega]+[\kappa,\omega]+\lambda_0, \qquad \dot\gamma=-\varepsilon\omega\gamma,
\end{equation}
where now $\kappa\in so(n)$ is a fixed matrix, $\mathbf k=\mathbf I\omega=\mathbb I\omega+D\omega$, $D=a^2m$,  and  $m$ is the mass of the system (ball with gyroscope).

After the substitution $\omega=\frac{1}{\varepsilon}\gamma\wedge\dot\gamma$, and taking the scalar product with $\frac{1}{\varepsilon}\gamma\wedge\delta\gamma$,
the equations \eqref{ORIGINALgyr} take the form
\begin{equation}\label{ORgyr}
\big\langle\frac{1}{\varepsilon^2}\mathbf I(\gamma\wedge\ddot\gamma)-\frac{1}{\varepsilon^3}[\mathbf I(\gamma\wedge\dot\gamma),\gamma\wedge\dot\gamma],\gamma\wedge\delta\gamma\big\rangle=\frac{1}{\varepsilon^2}\langle [\kappa,\gamma\wedge\dot\gamma],\gamma\wedge\delta\gamma\rangle,
\end{equation}
where we used that $\lambda_0$ is orthogonal to $\gamma\wedge \R^n$.
Now, since
\[
[\kappa,\gamma\wedge\dot\gamma]=(\kappa\gamma)\wedge\dot\gamma-(\kappa\dot\gamma)\wedge\gamma
\]
and
\[
\langle X\wedge Y,Z\wedge T\rangle=\langle X,Z\rangle\langle Y,T\rangle- \langle X,T\rangle\langle Y,Z\rangle,
\]
we get the right hand side of \eqref{ORgyr}:
\[
\mathrm{rhs}=\frac{1}{\varepsilon^2}\big(\langle \kappa\gamma,\gamma\rangle\langle\dot\gamma,\delta\gamma\rangle-\langle \kappa\gamma,\delta\gamma\rangle\langle\dot\gamma,\gamma\rangle
-\langle \kappa\dot\gamma,\gamma\rangle\langle\gamma,\delta\gamma\rangle+\langle \kappa\dot\gamma,\delta\gamma\rangle\langle\gamma,\gamma\rangle\big)
=\frac{1}{\varepsilon^2}\langle \kappa\dot\gamma,\delta\gamma\rangle.
\]
Similarly, the left hand side of \eqref{ORgyr} is given by
\[
\mathrm{lhs}=\Big\langle-\frac{1}{\varepsilon^2}\mathbf
I(\gamma\wedge\ddot\gamma)\gamma-\frac{1}{\varepsilon^3}\mathbf
I(\gamma\wedge\dot\gamma)\dot\gamma,\delta\gamma\Big\rangle=-\delta l+\mathbf{JK}(\dot\gamma,\delta\gamma),
\]
where  the second equality follows from \eqref{REDUCED}.
Therefore, from \eqref{ORgyr} we obtain

\begin{prop}\label{redukcija}
The reduced equations of motion of a gyroscopic ball rolling without slipping and twisting over a sphere are given by
\begin{equation}\label{REDUCED*}
\delta l-\mathbf{JK}(\dot\gamma,\delta\gamma)=\Big\langle\frac{1}{\varepsilon^2}\mathbf
I(\gamma\wedge\ddot\gamma)\gamma-\frac{1}{\varepsilon^3}\mathbf
I(\gamma\wedge\dot\gamma)\dot\gamma,\delta\gamma\Big\rangle=\mathbf f(\dot\gamma,\delta\gamma)
\end{equation}
where the gyroscopic term is given by
$\mathbf f(\dot\gamma,\delta\gamma)=\frac{1}{\varepsilon^2}\langle \dot\gamma,\kappa\delta\gamma\rangle$.
\end{prop}

Note that
the gyroscopic two-form $\mathbf f$
\begin{equation}\label{eq:f}
\mathbf f=\frac{1}{\varepsilon^2}\sum_{i<j}\kappa_{ij} d\gamma_i\wedge d\gamma_j.
\end{equation}
is exact  magnetic: $\mathbf f=d\mathbf a$, where
\[
\mathbf a=\frac{1}{2\varepsilon^2}\sum_{ij} \kappa_{ij}\gamma_id\gamma_j.
\]

Thus, the reduced equations of motion of a gyroscopic ball rolling without slipping and twisting over a sphere \eqref{REDUCED*} can be rewritten
in the equivalent form  (see Remark \ref{primedba}):
\[
\delta l_1=\mathbf{JK}(\dot\gamma,\delta\gamma),
\]
 where the Lagrangian $l_1$ is
\[
l_1(\dot\gamma,\gamma) =\frac{1}{2\varepsilon^2}\langle\mathbf
I(\gamma\wedge\dot\gamma) \dot\gamma,\gamma\rangle +\frac{1}{2\varepsilon^2}\langle \gamma,\kappa\dot\gamma\rangle.
\]

\begin{rem}
As in the 3-dimensional case, when $b$ tends to infinity, $\varepsilon$ tends to 1, $\gamma$ tends to the unit vector that is constant in the fixed reference frame
and we obtain the equations of motion of the Chaplygin ball with a gyroscope rolling without slipping and twisting over the plane orthogonal to $\gamma$.
\end{rem}

\begin{rem}
In addition, let us note that for $\varepsilon=1$ the system \eqref{ORIGINALgyr} with $\kappa=0$ represents also
the Veselova problem with the left-invariant metric on $SO(n)$ defined by the operator $\mathbf I$ (see \cite{VeVe2, FJ2004}).
 In this way, the system \eqref{ORIGINALgyr} for $\varepsilon=1$ can be seen as a Veselova problem with the addition of a gyroscope.
\end{rem}

Note that the Veselova problem is an example of an LR system.
These are nonholonomic systems with left-invariant metrics and right-invariant constraints on Lie groups \cite{VeVe2, FJ2004}.
One can consider LR systems with gyroscopic forces and their reduction to homogeneous spaces as well.
Along with the gyroscopic Chaplygin reduction, it is interesting to consider
the symplectic reduction of the corresponding Hamiltonian magnetic systems on Lie groups by using a general framework for the reduction
of the systems with symmetries on magnetic cotangent bundles given in \cite{KNP2005}.
The reduction problems based on \cite{KNP2005} will be consider elsewhere.

\subsection{Invariant measure} We are going to describe the reduced magnetic flow \eqref{REDUCED*} and its invariant measure on the cotangent bundle of a sphere $S^{n-1}$.
Consider the Legendre transformation of the Lagrangian $l$ given by \eqref{eq:l}.
\begin{equation}\label{Leg}
p=\frac{\partial l}{\partial \dot
\gamma}=\mathbf g(\dot\gamma)=-\frac{1}{\varepsilon^2}\mathbf
I(\gamma\wedge\dot\gamma)\gamma.
\end{equation}

Since $\mathbf I(\gamma\wedge\dot\gamma)$ is skew-symmetric, we get $\langle \gamma,p\rangle=0$. Thus,
the point $(p,\gamma)$ belongs to the cotangent bundle of a sphere
realized as a symplectic submanifold in the symplectic linear
space
$
(\R^{2n}\{\gamma, p\},dp_1\wedge
d\gamma_1+\cdots+dp_n\wedge d\gamma_n)
$
defined by the equations:
\begin{equation}\label{psi}
\phi_1=\langle \gamma,\gamma \rangle=1, \qquad \phi_2=\langle \gamma,p\rangle=0.
\end{equation}

Let
$
\dot\gamma=\mathbf g^{-1}(p)=X_\gamma(p,\gamma)
$
be the inverse of the Legendre transformation \eqref{Leg}, which is unique on the subvariety \eqref{psi}. Then
\begin{equation}\label{eq:h}
h(\gamma,p)=\frac12\langle X_\gamma(\gamma,p),p\rangle
\end{equation}
is the Hamiltonian function of the reduced system.
From \eqref{REDUCED} and \eqref{REDUCED*} we have
\[
\Big\langle-\dot p+\frac{1-\varepsilon}{\varepsilon^3}\mathbf
I(\gamma\wedge X_\gamma)X_\gamma,\delta\gamma\Big\rangle=\frac{1}{\varepsilon^2}\langle X_\gamma,\kappa\delta\gamma\rangle.
\]

Therefore
\[
\dot p=\frac{1-\varepsilon}{\varepsilon^3}\mathbf I(\gamma\wedge X_\gamma)X_\gamma+\frac{1}{\varepsilon^2}\kappa X_\gamma+\mu\gamma,
\]
where $\mu$ is the multiplier determined from the condition that $(\dot\gamma,\dot p)$ is tangent to $T^*S^{n-1}$:
\[
\langle \dot\gamma,p\rangle+\langle \gamma,\dot p\rangle=0.
\]

\begin{prop}\label{redukcija2}
The reduced equations of the
rolling of a ball with a gyroscope over a sphere without
slipping and twisting  on $T^*S^{n-1}$  are
\begin{equation}\label{dotG}
\dot\gamma =X_\gamma(\gamma,p), \qquad
\dot p =X_p(\gamma,p)=\frac{1-\varepsilon}{\varepsilon^3}\mathbf I(\gamma\wedge X_\gamma)X_\gamma+\frac{1}{\varepsilon^2}\kappa X_\gamma+\mu\gamma,
\end{equation}
where
\begin{equation}\label{Xp}
\mu=\frac{(\varepsilon-1)}{\varepsilon^3}\langle\left(\mathbf
I\left( \gamma\wedge X_\gamma\right)\right)X_\gamma,\gamma\rangle-2h(\gamma,p)+\frac{1}{\varepsilon^2}\langle X_\gamma,\kappa \gamma\rangle.
\end{equation}
\end{prop}

Let
\begin{equation}\label{eq:w}
\mathbf w=dp_1\wedge d\gamma_1+\dots+dp_n\wedge
d\gamma_n\,\vert_{T^*S^{n-1}}
\end{equation}
 be the canonical symplectic form on $T^*S^{n-1}$.

From Theorem \ref{novamera} and the formula for an invariant measure without magnetic term (see \cite{Jov2018b}), we have:

\begin{prop}\label{invarijantnaMera}
The reduced equations of the
rolling of a ball with a gyroscope over a sphere without
slipping and twisting \eqref{dotG} have an invariant measure $\nu(\gamma)\mathbf w^{n-1}$, where $w$ is from \eqref{eq:w} and $\nu$ is defined by:
\begin{equation}\label{IM}
\nu(\gamma):=(\det\mathbf I\vert_{\R^n\wedge\gamma})^{\frac{1}{2\varepsilon}-1}.
\end{equation}
\end{prop}

\section{Hamiltonization and integrability}\label{sec8}

\subsection{Hamiltonization of the $SO(n-2)$--invariant case}\label{sec8.1}
As already mentioned above, the existence of an invariant measure of a nonholonomic
system is closely related to the problem of its Hamiltonization. In this Section we provide a class of examples of $SO(n-2)$--symmetric systems
(ball with gyroscope) that  allow a Chaplygin Hamiltonization.

Consider the inertia operators
\begin{align}
\label{spec-op}
\mathbb I(\mathbf e_i\wedge \mathbf e_j)=(a_i a_j-D)\mathbf e_i\wedge
\mathbf e_j \quad \text{i.e.}, \quad  \mathbf I(X \wedge Y)=\mathbb AX\wedge
\mathbb AY,
\end{align}
parameterized by $\mathbb A=\diag(a_1,\dots,a_n)$, where $[\mathbf e_1, \dots, \mathbf e_n]$ is the standard basis of $\R^n$.
The formulas for the reduced Lagrangian $l$ \eqref{eq:l}, the Hamiltonian $h$ \eqref{eq:h}, and the density of an invariant measure $\nu$ \eqref{IM}
take the form:
\begin{align}
& l(\gamma,\dot\gamma)=\frac 1{2\varepsilon^2} \left(\langle \mathbb A\dot \gamma,\dot
\gamma\rangle\langle \mathbb A\gamma,\gamma\rangle- \langle \mathbb A\gamma,\dot \gamma\rangle^2 \right),
\label{lagrangian}\\
& h(\gamma,p)=\frac {\varepsilon^2}{2}  \frac{ \langle p, \mathbb A^{-1}p\rangle }{\langle \gamma,\mathbb A\gamma\rangle},
\label{hamiltonian}\\
& \nu(\gamma)=const\cdot\langle \mathbb A\gamma,\gamma\rangle^{\frac{n-2}{2\varepsilon}+2-n},
\label{measure}
\end{align}
(see \cite{Jov2015, Jov2018b}).
Moreover, the function
$
\N(\gamma)=\varepsilon\langle \mathbb A\gamma,\gamma\rangle ^{\frac{1}{2\varepsilon}-1}
$
is a Chaplygin multiplier: under the time substitution  $d\tau=\N(\gamma) dt$, the reduced system  \eqref{REDUCED}  with $\kappa=0$
becomes the geodesic flow of the metric
\begin{equation}
 ds^2_{A,\varepsilon}=(\gamma, A\gamma)^{\frac{1}{\varepsilon}-2}
\left((A d\gamma, d\gamma)(A\gamma,\gamma) -(A\gamma, d\gamma)^2
\right) \label{dsAe}
\end{equation}
defined by the Lagrangian (see \cite{Jov2018b})
\begin{equation} \label{L^*}
l^*(\gamma,\gamma^{\prime})=l(\gamma,\dot\gamma)\vert_{\dot\gamma=\N(\gamma)\gamma^\prime}=\frac12\langle \gamma,
\mathbb A\gamma\rangle^{\frac{1}{\varepsilon}-2}\left(\langle \mathbb A \gamma^{\prime},
\gamma^{\prime}\rangle\langle \mathbb A\gamma,\gamma\rangle  -\langle \mathbb A\gamma, \gamma^{\prime}\rangle^2
\right).
\end{equation}

\begin{rem}\label{3Doperator}
Note that for $n=3$ all symmetric operators $\mathbb I$ have the form \eqref{spec-op} in a basis formed by its eigenvectors.
Namely, after the  standard  identification $\R^3\cong so(3)$ \cite{Ar}, for the given inertia operator $\mathbb I=\diag(A,B,C)\colon \R^3\to\R^3$ for the gyroscopic ball and the
parameter $D=ma^2$, the operator $\mathbb I: so(3)\to so(3)$ has the form \eqref{spec-op}, with:
\begin{equation}\label{IandA}
\mathbb A=\diag\big(\frac{\Delta}{A+D},\frac{\Delta}{B+D},\frac{\Delta}{C+D}), \qquad \Delta=\sqrt{(A+D)(B+D)(C+D)}.
\end{equation}
\end{rem}

The above Hamiltonization
recovers the procedure of reduction and Hamiltonization  for a three-dimensional ball without gyroscope from \cite{EK2007}.
We would recall that Borisov and Mamaev proved the integrability  of the  three-dimensional ball  without gyroscope  and the spherical shell for a specific ratio between the radii: the case (iii) from Section \ref{sec6.1}, where $a=2b$, i.e. $\varepsilon=-1$, see \cite{BM2007}. The $n$--dimensional  reduced  system of a ball without gyroscope rolling over a sphere  \eqref{REDUCED} with the inertia operator $\mathbb I$ given by \eqref{spec-op}  is also integrable
for $\varepsilon=-1$; the integrability remains for such systems  for an arbitrary $\varepsilon$, if the matrix $\mathbb A$ has only two distinct
parameters  \cite{GajJov2019, GajJov2019b}.

Now, we turn to the systems with gyroscopic force.  If
\begin{equation}\label{condA*}
d(\N\mathbf f)=\frac{2}{\varepsilon}({\frac{1}{2\varepsilon}-1})\langle \mathbb A\gamma,\gamma\rangle ^{\frac{1}{2\varepsilon}-2}\big(\sum_k a_k\gamma_k d\gamma_k\big)
\wedge \big(\sum_{i<j}\kappa_{ij} d\gamma_i\wedge d\gamma_j\big)=0\vert_{T^*S^{n-1}}
\end{equation}
then the reduced gyroscopic system is Hamiltonizable as well. This follows from Theorem \ref{opsta}.

 For $n=3$, equation \eqref{condA*} is satisfied for an arbitrary gyroscopic term $\kappa$.
The following statement provides a class of examples, based on the $SO(n-2)$-symmetry, which satisfy  equation \eqref{condA*}, thus
are Hamiltonizable,  for every $n\ge 3$.

\begin{thm}\label{hamiltonizacija}
Assume that the gyroscopic term $\mathbf f$ from \eqref{eq:f} is given by $\kappa=\kappa_{12}\mathbf e_1\wedge\mathbf e_2$, i.e.,
\[
\mathbf f=\frac{\kappa_{12}}{\varepsilon^2} d\gamma_1\wedge d\gamma_2
\]
and the inertia operator of the system ball with gyroscope is given by
\eqref{spec-op}, where $a_3=a_4=\dots=a_n$:
\[
\mathbb A=\diag(a_1,a_2,a_3,\dots,a_3).
\]
Then the function $\N(\gamma)=\varepsilon\mathcal A(\gamma) ^{\frac{1}{2\varepsilon}-1}$, with
\begin{equation}\label{eq:A}
\mathcal A(\gamma)=a_3+(a_1-a_3)\gamma_1^2+(a_2-a_3)\gamma_2^2,
\end{equation}
 is a Chaplygin multiplier.
Under the time substitution  $d\tau=\N(\gamma) dt$ and the change of momenta $\tilde p=\N(\gamma)p$, the reduced system \eqref{dotG} becomes the
 magnetic geodesic flow of the metric \eqref{dsAe}
with respect to the twisted symplectic form given by
\begin{equation}
\tilde{\mathbf w}+\N\rho^*\mathbf f=d\tilde p\wedge d\gamma_1+\dots+d\tilde p_n \wedge d\gamma_n
+\frac{\kappa_{12}}{\varepsilon}\mathcal A(\gamma)^{\frac{1}{2\varepsilon}-1} d\gamma_1\wedge d\gamma_2\vert_{T^*S^{n-1}} .\label{TS}
\end{equation}
\end{thm}

\begin{rem}
The function \eqref{eq:A}  satisfies  $\mathcal A(\gamma)=\langle \mathbb A\gamma,\gamma\rangle $ for $\langle\gamma,\gamma\rangle=1$.
We use the function $\mathcal A$ to simplify some equations below.
For example, the Hamiltonian of the magnetic geodesic flow of the metric \eqref{dsAe} in the coordinates $(\gamma,\tilde p)$ can be written as
\begin{equation}
 h^*(\gamma,\tilde p)=\frac12 \mathcal A(\gamma)^{1-\frac{1}{\varepsilon}}{ \langle\tilde p, \mathbb A^{-1}\tilde p\rangle}.\label{eq:H}
\end{equation}
\end{rem}

\subsection{Integrability of the $SO(2)\times SO(n-2)$-invariant case}
In this section we want to impose additonal symmetry with respect to $SO(n-2)$-symmetry considered in Section \ref{sec8.1} and in particular
in Theorem \ref{hamiltonizacija}, This additional symmetry will imply integrability.

As mentioned above, the cotangent bundle $T^*S^{n-1}$  is realized within $\R^{2n}$
 by the constraints \eqref{psi}. In the new coordinates $(\gamma,\tilde p)=(\gamma,\varepsilon\mathcal A(\gamma) ^{\frac{1}{2\varepsilon}-1} p)$, the constraints become
\begin{equation}\label{phi*}
\phi^*_1=\langle \gamma,\gamma \rangle=1, \qquad \phi^*_2=\frac{1}{\varepsilon}\mathcal A(\gamma) ^{1-\frac{1}{2\varepsilon}}\langle \gamma,\tilde p\rangle=0.
\end{equation}
Instead of \eqref{phi*}, we equivalently use  the constraints
\begin{equation}
\psi_1= \langle\gamma,\gamma\rangle=1, \qquad \psi_2= \langle\tilde p,\gamma\rangle=0.
\label{psi*}
\end{equation}

The magnetic Poisson bracket on the cotangent bundle $T^*S^{n-1}\subset \R^{2n}\{\gamma,\tilde p\}$
can be described by the Dirac construction as follows:
\[
\{F, G\}_d =\{F, G\}^\kappa -\frac{\{F,\psi_1 \}^\kappa\{G,\psi_2\}^\kappa- \{F,\psi_2\}^\kappa
\{G,\psi_1\}^\kappa }{ \{\psi_1, \psi_2 \}^\kappa },
\]
where
\[
\{F,G\}^\kappa=\{F,G\}^0+
\frac{\kappa_{12}}{\varepsilon}\mathcal A(\gamma)^{\frac{1}{2\varepsilon}-1}\big(\frac{\partial F}{\partial\tilde p_1}
\frac{\partial G}{\partial\tilde p_2}-\frac{\partial F}{\partial\tilde p_2}\frac{\partial G}{\partial\tilde p_1}\big)
\]
and
\[
\{F,G\}^0=\sum_{i=1}^n\big(\frac{\partial
F}{\partial\gamma_i} \frac{\partial G}{\partial\tilde
p_i}-\frac{\partial F}{\partial\tilde p_i}\frac{\partial
G}{\partial\gamma_i}\big)
\]
is the canonical Poisson bracket on $\R^{2n}\{\gamma,\tilde p\}$, (see \cite{AKN}).
Considered on $\R^{2n}\{\gamma,\tilde p\}$ without the subset $\{\gamma=0\}$, the bracket $\{\cdot,\cdot\}_d$ is
degenerate with two Casimir functions $\psi_1$ and $\psi_2$.
The symplectic leaf given by \eqref{psi*} is exactly the cotangent
bundle $T^*S^{n-1}$ endowed with the twisted symplectic form \eqref{TS}.

It is convenient to derive the equations of the magnetic Hamiltonian flows with respect to the Dirac bracket $\{\cdot,\cdot\}_d$
using  the Lagrange multipliers and the magnetic Hamiltonian flows with respect to the magnetic bracket $\{\cdot,\cdot\}^\kappa$   (e.g., see \cite{AKN}). Let
\begin{equation*}
H=h^*-\lambda_1 \psi_1-\lambda_2\psi_2.
\end{equation*}

The magnetic Hamiltonian flow generated by the Hamiltonian  \eqref{eq:H}
with respect to the Dirac bracket $\{\cdot,\cdot\}_d$ is given by
\begin{align}
\label{eq:HamSys}
\gamma^\prime=&\frac{\partial H}{\partial \tilde p}= \mathcal A(\gamma)^{1-\frac{1}{\varepsilon}}\A^{-1}\tilde p-\lambda_2
\gamma,\\
\label{eq:HamSys2}
\tilde p^\prime=&-\frac{\partial H}{\partial \gamma}
 +\frac{\kappa_{12}}{\varepsilon}\mathcal A(\gamma)^{\frac{1}{2\varepsilon}-1} \mathbf e_1 \wedge \mathbf e_2(\gamma^\prime) \\
\nonumber =& \frac{1-\varepsilon}{\varepsilon}\mathcal A(\gamma)^{-\frac{1}{\varepsilon}}\langle\tilde p, \A^{-1}\tilde
p\rangle \big((a_1-a_3)\gamma_1 \mathbf e_1+(a_2-a_3)\gamma_2\mathbf e_2\big)+2\lambda_1 \gamma+\lambda_2 \tilde p\\
\nonumber &+\frac{\kappa_{12}}{\varepsilon}
\mathcal A(\gamma)^{\frac{1}{2\varepsilon}-1}\big((\mathcal A(\gamma)^{1-\frac{1}{\varepsilon}}\frac{\tilde p_2}{a_2}-\lambda_2
\gamma_2)\mathbf e_1-(\mathcal A(\gamma)^{1-\frac{1}{\varepsilon}}\frac{\tilde p_1}{a_1}-\lambda_2
\gamma_1)\mathbf e_2 \big),
\end{align}
where the Lagrange multipliers $\lambda_1$ and $\lambda_2$  are determined  from the condition that the functions $\psi_1$ and $\psi_2$ are integrals of the flow.

From now on we consider the system \eqref{eq:HamSys}, \eqref{eq:HamSys2} restricted to the
symplectic leaf \eqref{psi*}, that is, we consider the magnetic geodesic flow of the metric \eqref{dsAe}.

Let us impose now the additional symmetry.
Suppose: $a_1=a_2\ne a_3$.  Both the Hamiltonian \eqref{eq:H} and the magnetic two-form \eqref{TS} are invariant with respect to the action of the group $SO(2)\times SO(n-2)$. We first consider the case $\kappa_{12}=0$:
the corresponding first integrals are linear and given as follows:
\[
\Phi^0_{12}=\gamma_1 \tilde p_2-\gamma_2\tilde p_1, \quad
\Phi^0_{ij}=\gamma_i\tilde p_j-\gamma_j\tilde p_i, \quad 3\le i<j\le n.
\]
Such first integrals are sometimes called {\it Noether integrals} as their existence follow from the Emmy Noether theorem.
Let us now consider a general case $\kappa_{12}\ne 0$:
straightforward calculations show that $\Phi_{ij}=\Phi_{ij}^0$, $3\le i<j\le n$ remain to be first integrals for $\kappa_{12}\ne 0$. Moreover,
\[
\frac{d}{d\tau}\Phi^0_{12}=-\frac{\kappa_{12}}{\varepsilon}\mathcal A(\gamma) ^{\frac{1}{2\varepsilon}-1}(\gamma_1\gamma_1'+\gamma_2\gamma_2')
=-\frac{\kappa_{12}}{a_1-a_3}\frac{d}{d\tau}\big(\mathcal A(\gamma)^{\frac{1}{2\varepsilon}}\big).
\]

Thus, the first integrals for $\kappa_{12}\ne 0$ are
\[
\Phi_{12}=\gamma_1 \tilde p_2-\gamma_2\tilde p_1+\frac{\kappa_{12}}{a_1-a_3}\mathcal A(\gamma)^{\frac{1}{2\varepsilon}},
\quad
\Phi_{ij}=\gamma_i\tilde p_j-\gamma_j\tilde p_i, \quad 3\le i<j\le n.
\]
These first integrals are the components of the momentum mapping of the $SO(2)\times SO(n-2)$-action with respect to the twisted symplectic form \eqref{TS}.

\begin{thm}\label{integrabilni}
For $a_1=a_2\ne a_3$ the magnetic geodesic flow of the metric $ds^2_{\mathbb A,\varepsilon}$ defined by the Hamiltonian \eqref{eq:H}
with respect to the twisted symplectic form \eqref{TS} is completely integrable.

\begin{itemize}
\item[(i)] If $n=3$ the system is Liouville integrable. Generic invariant manifolds are two-dimensional Lagrangian tori, the common level sets of $h^*$ and $\Phi_{12}$.

\item[(ii)]  If $n=4$ the system is Liouville integrable. Generic invariant manifolds are three-dimensional Lagrangian tori, the common level sets of $h^*$, $\Phi_{12}$, and $\Phi_{34}$.

\item[(iii)]  If $n\ge 5$ the system is integrable in the noncommutative sense. Generic invariant manifolds are three-dimensional isotropic tori, the common level sets of $h^*$, $\Phi_{12}$, and $\Phi_{ij}$, $3\le i<j\le n$.
    \end{itemize}
\end{thm}

\begin{proof}
 For $n=3$ the statement is clear.
For $n=4$, the Hamiltonian system \eqref{eq:HamSys} possesses three independent  integrals $h^*$, $\Phi_{12}$, $\Phi_{34}$, in involution:
\[
\{h^*,\Phi_{12}\}_d=0, \quad \{h^*,\Phi_{34}\}_d=0,\quad  \{\Phi_{12},\Phi_{34}\}_d=0.
\]
Thus, the Hamiltonian system \eqref{eq:HamSys}, \eqref{eq:HamSys2} is
completely integrable according to the Arnold-Liouville theorem.

For $n>4$, generic common level sets of all integrals are three-dimensional tori as well.
Indeed, consider the natural embedding $T^*S^3 \subset T^*S^{n-1}$ induced by the embedding
$\mathrm{span}\{\mathbf e_1,\mathbf e_2,\mathbf e_3,\mathbf e_4\}\subset \R^n$.
Let us set
$\mathbf P=(\tilde p_3,\tilde p_4,\dots,\tilde p_n)$, $\Gamma=(\gamma_3,\gamma_4,\dots,\gamma_n)$.
Then $\tilde p=(\tilde p_1,\tilde p_2,\mathbf P)$, $\gamma=(\gamma_1,\gamma_2,\Gamma)$.

The system \eqref{eq:HamSys}, \eqref{eq:HamSys2} is invariant with respect to the $SO(n-2)$--action
\[
R(\gamma,\tilde p)=(\gamma_1,\gamma_2,R\Gamma,\tilde p_1,\tilde p_2,R\mathbf P), \quad R\in SO(n-2).
\]
Also, as we already mentioned, the integrals $\Phi_{ij}$, $3\le i<j\le n$ are components of the corresponding momentum mapping
\[
(\gamma,\tilde p)\longmapsto \Gamma\wedge\mathbf P.
\]

For any point  $\mathbf c_0=(\gamma_0,\tilde p_0)\in T^*S^{n-1}$, there exists a matrix $R_0\in SO(n-2)$, such that
$\mathbf d_0=R_0(\gamma_0,\tilde p_0)$ belongs to $T^*S^3$. Since the system is invariant with respect to the $SO(n-2)$--action,
the solution $\mathbf c(\tau)=(\gamma(\tau),\tilde p(\tau))$ with the initial condition
$\mathbf c(0)=(\gamma(0),\tilde p(0))=\mathbf c_0$
is mapped to the solution $\mathbf d(\tau)=R(\gamma(\tau),\tilde p(\tau))$ with the initial condition
\[
\mathbf d(0)=R_0(\gamma(0),\tilde p(0))=R_0(\gamma_0,\tilde p_0)=\mathbf d_0\in T^*S^3.
\]

The solutions $\mathbf c(\tau)$ and $\mathbf d(\tau)$ have the same energy, $h^*(\mathbf c_0)=h^*(\mathbf d_0)$, while the corresponding values of the momenta are different:
the momentum of $\mathbf c(\tau)$ is transformed to the momentum of $\mathbf d(\tau)$ by the adjoint mapping
\[
\Gamma_0\wedge \mathbf P_0 \longmapsto R_0 (\Gamma_0\wedge\mathbf P_0) R^T_0=\Phi_{34}(\mathbf d_0) \mathbf e_3 \wedge \mathbf e_4,
\]
where $\mathbf c_0=(\gamma_{0,1},\gamma_{0,2},\Gamma_0,\tilde p_{0,1},\tilde p_{0,2},\mathbf P_0)$.

One can easily verify that the solution $\mathbf d(\tau)$ belongs $T^*S^3$, that is, it is a solution of the problem for $n=4$. Therefore, generically,
$\mathbf d(\tau)$ is a quasi-periodic trajectory over a 3--dimensional invariant torus $\mathcal T_0\subset T^*S^3$, the connected component of the level set
\[
h^*=h^*(\mathbf d_0), \qquad \Phi_{12}=\Phi_{12}(\mathbf d_0), \qquad \Phi_{34}=\Phi_{34}(\mathbf d_0).
\]
All other components of the momentum mapping $\Phi_{ij}$, $3\le i<j\le n$, $(i,j)\ne (3,4)$ are equal to zero.

Note that a point $\mathbf d\in T^*S^{n-1}$ belongs to
$T^*S^3$ if and only if 
\[
\Phi_{ij}(\mathbf d)=0, \qquad 3\le i<j\le n, \qquad (i,j)\ne (3,4).
\]
Thus, the original trajectory $\mathbf c(\tau)=R^{-1}_0(\mathbf d(\tau))$ is quasi-periodic over the 3-dimensional invariant torus $\mathcal T=R_0^{-1}(\mathcal T_0)$,
which is also the connected component of the level set
\[
h^*=h^*(\mathbf c_0)=h^*(\mathbf d_0), \qquad \Phi_{12}=\Phi_{12}(\mathbf c_0), \qquad \Phi_{ij}=\Phi_{ij}(\mathbf c_0), \quad 3\le i < j \le n.
\]

The integrability of the system is a particular example of so-called noncommutative integrability.
Namely, since the common level sets
of the integrals are 3--dimensional, and
the Hamiltonian system \eqref{eq:HamSys}, \eqref{eq:HamSys2} has three  independent first integrals $h^*$, $\Phi_{12}^\kappa$, and
$
\sum_{3\le i<j\le n}(\Phi_{ij})^2,
$
that commute with all integrals,  the system
is completely integrable according the Nekhoroshev-Mishchenko-Fomenko theorem on non-commutative integrability for all $n>4$ (e.g., see \cite{AKN}).\end{proof}

Note that in the original phase space $T^*S^{n-1}\{\gamma,p\}$, the first integrals have the form
\[
\Phi_{12}=\varepsilon\mathcal A(\gamma) ^{\frac{1}{2\varepsilon}-1}(\gamma_1 p_2-\gamma_2p_1)+\frac{\kappa_{12}}{a_1-a_3}\mathcal A(\gamma)^{\frac{1}{2\varepsilon}}, \]
and
\[
\Phi_{ij}=\varepsilon\mathcal A(\gamma) ^{\frac{1}{2\varepsilon}-1}(\gamma_i p_j-\gamma_j p_i), \qquad 3\le i<j\le n.
\]
In the original time, the system over a regular invariant torus $\mathcal T$ has the form \eqref{Jacobi}, where $\Phi = \N^{-1}\vert_\mathcal T$.

\begin{rem}
For $n=3$,
within the standard isomorphism between Lie algebras $(so(3),[\cdot,\cdot])$ and $(\R^3,\times)$
given by
\begin{equation}\label{izomorfizam}
a_{ij}=-\varepsilon_{ijk}a_k, \qquad i,j,k=1,2,3
\end{equation}
(see \cite{Ar}),
the equations \eqref{ORIGINALgyr} with the inertia operator defined by \eqref{spec-op}, $\mathbb A=\diag(a_1,a_1,a_3)$, and
$\kappa=\kappa_{12}\mathbf e_1\wedge \mathbf e_2$ correspond to
the equations \eqref{e3a}
defined by the Bobilev conditions \eqref{bobiljev} with $\vec\kappa=-\kappa_{12}\vec{\mathbf e}_3$ and $\mathbb I$ and $\mathbb A$ related by \eqref{IandA}  (see Subsection \ref{sec6.4} and Remark \ref{3Doperator}).
Then, along with the Liouville integrability after the Hamiltonization described in Theorem \ref{integrabilni},
the  system is also integrable according to the Euler-Jacobi theorem.
\end{rem}

\section{Generalized Demchenko case without twisting in $\R^n$}\label{sec9}

\subsection{Definition of the system}
As above, we will consider the rolling of a gyroscopic ball $\mathbf B$ without slipping and twisting in $\R^n$, now with an additional symmetry of the system.
The additional symmetry is analogous to the Zhukovskiy condition \eqref{uslovZh} in dimension $n=3$. Recall,  that adding a gyroscopic term does not change formulas for curvature of the distribution $\mathcal D$, ${\bf JK}$ term \eqref{7.5a} and $\Sigma$ term \eqref{7.4}. For the curvature $K$ of $\mathcal D$ see Lemma 7 in \cite{Jov2018b}:
\[
K_{(g,\mathbf{r})}(\xi_1^h,\xi_2^h)=\frac{2\varepsilon-1}{\varepsilon^2}\Ad_g(\xi_1\wedge\xi_2),  \qquad
\xi_1,\xi_2\in T_{\pi(g,\mathbf{r})} S^{n-1}.
\]

Since the reduced gyroscopic form $\mathbf f$ is exact magnetic for an arbitrary $\kappa\in so(n)$,
\begin{equation}\label{kapa}
\kappa=\sum_{i<j} \kappa_{ij} \mathbf e_i \wedge \mathbf e_j,
\end{equation}
if the {\bf JK}-term in \eqref{ChaplyginRed} vanishes, then the reduced gyroscopic $G$--Chaplygin system \eqref{ChaplyginRed} is automatically Hamiltonian
without any time reparametrization.

We provide two situations when such conditions are satisfied, for the
rolling of a gyroscopic Chaplygin ball without slipping and twisting over a sphere $S^{n-1}$
 (see Remark \ref{nemaReakcije}).
The first situation: if the radii of the sphere and the ball are equal, which is equivalent to the condition $\varepsilon=1/2$, then the curvature $K$ of $\mathcal D$ vanishes
(the constraints are holonomic).
 Since the {\bf JK}-term is given by the coupling of the curvature $K$ with the momentum mapping of the $SO(n)$--action on the configuration
space \eqref{principal} (see Remark \ref{JKclan}),
we have ${\bf JK}=0$.
The second situation we get when the  inertia operator $\mathbb I$ of the system, that is, the modified inertia operator $\mathbf I$,   is proportional to the identity operator. Then   the coupling between the curvature and the momentum mapping vanishes, see \eqref{7.5a}, although the curvature
of $\mathcal D$ is different from zero. Let us remind that the curvature of the distribution measure the nonholonomicity of the constraints: it is zero if and only if the constraints are holonomic.

These two situations do not require a time reparametrisation for a Hamiltonization: the reduced equations \eqref{dotG} are Hamiltonian with respect to the symplectic form $\mathbf w+\rho^*\mathbf f$, where $\mathbf w$ is the canonical symplectic form \eqref{eq:w}.

For $n=3$, the condition that the inertia operator $\mathbb I$ is proportional to the identity operator is equivalent to the Zhukovskiy condition \eqref{uslovZh}. One gets the case of motion of a gyroscopic ball considered by Demchenko in \cite{Dem1924}, see also \cite{DGJ} and subsection \ref{ex:Dem} below, under an additional non-twisting condition.  This motivates us to introduce the following definition of a generalized
Demchenko case without twisting in higher dimensions.

\begin{dfn}\label{def:GenDem}
We say that the ball with a gyroscope satisfies the \emph{Zhukovskiy condition} if the
inertia operator $\mathbb I$ of the system is proportional to the identity operator.
The \emph{generalized Demchenko case without twisting in} $\R^n, \, n\geq 3$, is a system of a balanced $n$-dimensional
gyroscopic  ball satisfying the Zhukovskiy condition, rolling without slipping and twisting over a fixed $(n-1)$--dimensional sphere.
\end{dfn}

As before, we consider the cotangent bundle $T^*S^{n-1}\subset \R^{2n}\{\gamma,p\}$
realized by the constraints \eqref{psi},
$\mathbf w$ is the canonical symplectic form on $T^*S^{n-1}$ given with \eqref{eq:w} and $\rho$ is the canonical projection $\rho: T^*S^{n-1}\to S^{n-1}$.
 Now, the magnetic Poisson brackets on $\R^{2n}\{\gamma,p\}$ without the set $\{\gamma=0\}$ are defined by:
\begin{equation}\label{poisson}
\{F, G\}_d =\{F, G\}^\kappa -\frac{\{F,\phi_1 \}^\kappa\{G,\phi_2\}^\kappa- \{F,\phi_2\}^\kappa
\{G,\phi_1\}^\kappa }{ \{\phi_1, \phi_2 \}^\kappa },
\end{equation}
where
\[
\{F,G\}^\kappa=\sum_{i}\big(\frac{\partial
F}{\partial\gamma_i} \frac{\partial G}{\partial
p_i}-\frac{\partial F}{\partial p_i}\frac{\partial
G}{\partial\gamma_i}\big)+
\frac{1}{\varepsilon^2}\sum_{i,j} \kappa_{ij}\frac{\partial
F}{\partial p_i} \frac{\partial G}{\partial p_j}
\]
and $\phi_1, \phi_2$ are given in \eqref{psi}.
The symplectic leaf given by \eqref{psi} is the cotangent
bundle $T^*S^{n-1}$ endowed with the twisted symplectic form $\mathbf w+\rho^*\mathbf f$.

Let the modified inertia operator $\mathbf I=\mathbb I+ D\mathrm{Id}_{so(n)}$ $(D=m a^2)$  be equal to the identity operator on $so(n)$ multiplied by a constant $\tau$. For example, we can take $\mathbb I$ given by \eqref{spec-op} with $\mathbb A=\diag(\sqrt{\tau},\dots,\sqrt{\tau})$.
Then the reduced Hamiltonian takes the form
\begin{equation}\label{DemHam}
h=\frac{\varepsilon^2}{2\tau}\langle p,p\rangle.
\end{equation}

By taking $H=h-\lambda_1\phi_1-\lambda_2\phi_2$, we obtain
the magnetic Hamiltonian flow of the Hamiltonian \eqref{DemHam} with respect to the Dirac bracket \eqref{poisson}
\begin{align}
\label{HFmag1}\dot\gamma=& \frac{\partial H}{\partial p}= \frac{\varepsilon^2}{\tau} p-\lambda_2 \gamma,\\
\label{HFmag2} \dot p   =&  -\frac{\partial H}{\partial \gamma}+ \frac{1}{\varepsilon^2}\kappa\big(\frac{\partial H}{\partial p}\big)=
2\lambda_1\gamma+\lambda_2 p+\frac{1}{\tau}\kappa p -\frac{\lambda_2}{\varepsilon^2} \kappa \gamma.
\end{align}
Here, from the condition that $\phi_1$ and $\phi_2$ are first integrals
of the flow, the Lagrange multipliers can be calculated to get
\[
\lambda_1=\frac{\frac{1}{\tau}\langle p,\kappa\gamma \rangle-\frac{\varepsilon^2}{\tau}\langle p,p\rangle}{2\langle \gamma,\gamma\rangle },
\qquad \lambda_2=\frac{\epsilon^2}{\tau}\frac{\langle p,\gamma\rangle}{\langle \gamma,\gamma\rangle}.
\]

 \begin{prop} The equations of motion of the $n$-dimensional generalized Demchenko case without twisting are:
\begin{equation}\label{DEMCENKOgyr}
 \tau\dot{\omega}=[\kappa,\omega]+\lambda_0, \qquad \dot\gamma=-\varepsilon\omega\gamma,
\end{equation}
where $\kappa\in so(n)$ is a fixed skew-symmetric matrix  \eqref{kapa}
and the Lagrange multiplier $\lambda_0\in(\R^n\wedge\gamma)^\perp$ is determined from the condition that
$\omega\in\R^n\wedge\gamma$.
The equations of motion
reduce to the magnetic geodesic flow of the Hamiltonian \eqref{DemHam} with respect to the bracket \eqref{poisson}
\begin{equation}\label{dem}
\dot\gamma=\frac{\varepsilon^2}{\tau} p,\qquad
\dot p=\frac{1}{\tau}\kappa p+\mu\gamma, \qquad \mu=\frac{1}{\tau}\langle p,\kappa\gamma \rangle-\frac{\varepsilon^2}{\tau}\langle p,p\rangle,
\end{equation}
restricted to the cotangent bundle of the sphere \eqref{psi}.
\end{prop}

The proof follows from \eqref{ORIGINALgyr},  the equations \eqref{HFmag1}, \eqref{HFmag2} restricted to \eqref{psi},  and Proposition \ref{redukcija2}.

When $\varepsilon=1$, we obtain the equations of motion of a gyroscopic ball  rolling without slipping and twisting  over the plane orthogonal to $\gamma$, such that the
the inertia operator $\mathbb I$ of the system is proportional to the identity operator. In dimension $n=3$ this is the Zhukovskiy problem with an additional nontwisting condition (see Section \ref{sec6}).

Let us note that integrable magnetic Hamiltonian systems on $S^2$  were  studied in  \cite{S2002},   using their
relation to a special Neumann system on $S^3$.
In particular,  the reduced problem \eqref{dem} for $n=3$ was described there
by using the Cartan model of the sphere $S^2$ within the group $SU(2)$.
Although the systems \eqref{dem} are quite natural as they are described by the round metric on a sphere with a magnetic field defined by a constant
two-form in the ambient space, they have not been studied before for $n>3$.

Since $\mathbb I$ (and equivalently $\mathbf I$) is proportional to the identity matrix, we can consider,  without loss of generality,
the system  in a suitable orthonormal basis $[\mathbf{e}_1,\dots, \mathbf{e}_n]$  of $\R^n$,  such that
the skew-symmetric matrix  \eqref{kapa}  takes the form
\[
\kappa=\kappa_{12} \mathbf e_1\wedge \mathbf e_2+\kappa_{34} \mathbf e_3\wedge \mathbf e_4 +\dots +\mathbf \kappa_{2[n/2]-1,2[n/2]}\mathbf e_{2[n/2]-1}\wedge\mathbf e_{2[n/2]}.
\]

\subsection{Three-dimensional Demchenko case without twisting\label{ex:Dem}} In his PhD thesis \cite{Dem1924} (see also \cite{DGJ}) Demchenko studied the rolling of a
 ball with a gyroscope without slipping over a fixed sphere in $\mathbb{R}^3$. He assumed that the ball is  dynamically  axially symmetric, that axis of gyroscope coincide with symmetry axis of the ball, and that the inertia operators of the ball  and the gyroscope  satisfy the Zhukovskiy condition \eqref{uslovZh},  that is, the inertia operator of the system is proportional to the identity matrix:
 $\mathbb I=\diag(A,A,A)$.

The equations of motion are (see \eqref{Chap1})
\begin{equation}
\dot{\vec{\mathbf k}}=(\vec{{\mathbf k}}+\vec{\kappa})\times\vec{\omega},
\qquad \dot{\vec{\gamma}}=\varepsilon \vec{\gamma}\times\vec{\omega},
\label{demoriginal}
\end{equation}
 where $\vec{{\mathbf  k}}=(A+ma^2)\vec{\omega}-ma^2 \langle \vec{\omega},\vec{\gamma}\rangle \vec{\gamma}$.
 Demchenko  solved the system via elliptic functions.

Now, we add the no-twisting condition on the Demchenko rolling, e.q. we additionally assume that the angular velocity $\vec{\omega}$ belongs to the common tangent plane of the ball and the sphere in their contact point. The equations of motion are (see \eqref{e3a})
\begin{equation}\label{3dem}
\dot{\vec{\mathbf k}}=\vec{\kappa}\times\vec{\omega}+\lambda\vec\gamma,\qquad \vec{\gamma}=\varepsilon\vec{\gamma}\times\vec{\omega},
\end{equation}
where  $\vec{\mathbf{k}}=(A+ma^2)\vec{\omega}=((A+ma^2)\omega_1, (A+ma^2)\omega_2, (A+ma^2)\omega_3))$  and
$\lambda$  is the Lagrange multiplier of the constraint $\langle\vec{\omega},\vec{\gamma}\rangle=0$,
\[
\lambda=-\langle \vec\gamma,\vec\kappa\times \vec\omega\rangle.
\]

After the identification \eqref{izomorfizam}, the matrix system
\eqref{DEMCENKOgyr}, for $n=3$, becomes the system \eqref{3dem} in the vector notation, where the matrix multiplier $\lambda_0$ corresponds to  $\lambda\vec\gamma$,
 $\gamma\equiv \vec\gamma$, and the parameter $\tau$ is equal to $A+ma^2$ (see Remark \ref{3Doperator}).

The reduced equations of motion \eqref{dem} on $T^*S^2$,
for
$
\kappa=\kappa_{12}\mathbf e_1\wedge \mathbf e_2,
$
become
\begin{equation}\label{3dred}
\begin{aligned}
{\dot\gamma}_1&=\frac{\varepsilon^2}{\tau}p_1,\qquad\qquad {\dot p}_1=\frac{1}{\tau}\kappa_{12} p_2+\mu\gamma_1,\\
{\dot\gamma}_2&=\frac{\varepsilon^2}{\tau}p_2,\qquad\qquad {\dot p}_2=-\frac{1}{\tau}\kappa_{12} p_1+\mu\gamma_2,\\
{\dot\gamma}_3&=\frac{\varepsilon^2}{\tau}p_3,\qquad\qquad {\dot p}_3=\mu\gamma_3,\\
\mu&=\frac{\kappa_{12}}{\tau}  (p_1\gamma_2-p_2\gamma_1)-\frac{\varepsilon^2}{\tau}(p_1^2+p_2^2+p_3^2),
\end{aligned}
\end{equation}
They are Hamiltonian with respect to the Poisson structure \eqref{poisson}  and the  Hamiltonian is
\[
h=\frac{\varepsilon^2}{2\tau}(p_1^2+p_2^2+p_3^2).
\]

\begin{thm}
The reduced equations of  the  Demchenko case without twisting  \eqref{3dred}
are Liouville integrable on $T^*S^2$
with the first integrals $h$, $\Phi$, where
$$
\Phi(\gamma,p)=\gamma_1 p_2-\gamma_2p_1+\frac{\kappa_{12}}{2\varepsilon^2}(\gamma_1^2+\gamma_2^2).
$$
\end{thm}
Proof follows by a direct calculation.

The reduced     system \eqref{3dred}   can be solved in elliptic quadratures.

\begin{thm}\label{th:Integration3d}
The reduced equations of the three-dimensional Demchenko case without twisting \eqref{3dred} can be explicitly integrated via elliptic functions and their degenerations.
\end{thm}
\begin{proof}
 Instead on the cotangent bundle $T^*S^2\{\gamma,p\}$,  we will equivalently integrate the system
on the tangent bundle $TS^2\{\gamma,\dot\gamma\}$.
Let us introduce polar coordinates $r, \varphi$ by
\[
\gamma_1=r\cos\varphi,\ \gamma_2=r\sin\varphi.
\]
From the condition $\langle\gamma,\gamma\rangle=1$ it follows that $r^2+\gamma_3^2=1$, while $\langle\gamma,\dot\gamma\rangle=0$ is identically satisfied.
 By differentiating $r^2+\gamma_3^2=1$ with respect to time, one gets $\dot{\gamma}_3^2=\frac{r^2}{1-r^2}\dot{r}^2$.

In the new coordinates, using the last relation, the first integrals can be rewritten as:
 \begin{align}
h&=\frac{\tau}{2\varepsilon^2}\big(\dot{r}^2+r^2\dot{\varphi}^2+\frac{r^2\dot{r}^2}{1-r^2}),\label{h3}\\
\Phi&=\frac{\tau}{\varepsilon^2}r^2\dot{\varphi}+\frac{\kappa_{12}}{2\varepsilon^2}r^2\label{f3}.
\end{align}

Note that $\tau>0$. We also assume $h>0$ since $h=0$ corresponds to the equilibrium positions.

From \eqref{f3},  we get
\begin{equation}\label{fi}
\dot{\varphi}=\frac{2\varepsilon^2\Phi-\kappa_{12} r^2}{2\tau r^2},\
\end{equation}
and, by plugging into \eqref{h3},  it follows
\[
\dot{r}^2=\big(\frac{\varepsilon^2}{\tau^2}(2h\tau+\kappa_{12}\Phi)-\frac{\kappa_{12}^2}{4\tau^2}r^2-\frac{\varepsilon^4\Phi^2}{\tau^2}\frac{1}{r^2}\big)(1-r^2).
\]

Introducing $u=r^2$, one derives
\begin{align}
\nonumber            &  \qquad\qquad\qquad \dot{u}^2=Q_3(u), \\
\label{eq:Q}   Q_3(u) :&=\frac{\kappa_{12}^2}{\tau^2}(u-1)\big(u^2-\frac{4\varepsilon^2}{\kappa_{12}^2}(2h\tau+\kappa_{12}\Phi)u+\frac{4\varepsilon^4\Phi^2}{\kappa_{12}^2}\big)\\
\nonumber            &=\frac{\kappa_{12}^2}{\tau^2}(u-1)(u-u_1)(u-u_2).
\end{align}

Thus, $r^2$ can be expressed as an elliptic function (or its degenerations) of time. Using $\gamma_3^2=1-r^2$, one gets $\gamma_3$, and from \eqref{fi} one finds $\varphi$ after an integration.\end{proof}

Notice that the polynomial $Q_3$ \eqref{eq:Q} always has $u=1$ as a root. Observe also:
\[
Q_3(0)=-\frac{ 4 \varepsilon^4\Phi^2}{\tau^2}<0.
\]

From Vieta's formulas, it follows that $u_1u_2>0$, or in other words, the remaining two roots $u_1, u_2$ of $Q_3$ are of the same sign. Having in mind that $0\leqslant u\leqslant 1$, the real solutions,  for $u_1< u_2$,  corresponds to the following cases:

\begin{itemize}

\item[(A)] $0<u_1<u_2<1$; Case  (A)  happens when the discriminant of the polynomial $Q_2(u)=(u-u_1)(u-u_2)$ is greater then zero,
 the minimum of $Q_2(u)$ is between  $0$ and $1$, and  $Q_2(1)>0$. This yields
conditions:
\[
\begin{aligned}
&h\tau+\kappa_{12}\Phi>0,\\
&2h\tau+\kappa_{12}\Phi<\frac{\kappa_{12}^2}{2\varepsilon^2},\\
&2h\tau+\kappa_{12}\Phi-\varepsilon^2\Phi<\frac{\kappa_{12}^2}{4\varepsilon^2}
\end{aligned}
\]

\item[(B)] $0<u_1<1<u_2$. Case (B)  happens when $Q_2(1)<0$,  that is
\[
2h\tau+\kappa_{12}\Phi-\varepsilon^2\Phi>\frac{\kappa_{12}^2}{4\varepsilon^2}
\]

\end{itemize}

In both cases $r$ belongs to an annulus:
\[
\text{Case (A)} \quad \sqrt{u_1}\leqslant r\leqslant\sqrt{u_2};\qquad  \text{Case (B)}\quad \sqrt{u_1}\leqslant r\leqslant 1.
\]

When the discriminant of the polynomial $Q_3$ \eqref{eq:Q}  vanishes, the corresponding elliptic functions degenerate. It happens  if
 $u_1=u_2$, or when one of the roots $u_1, u_2$ is equal to $1$. Direct calculations show that the
discriminant of the polynomial $Q_3$ vanishes when
\[
h\tau+\kappa_{12}\Phi=0, \qquad \text{or} \qquad 2h\tau+\kappa_{12}\Phi -  \varepsilon^2\Phi=\frac{\kappa_{12}^2}{4\varepsilon^2}.
\]
The first case corresponds to the condition that the discriminant of $Q_2$
is zero, and  the second case corresponds to  $Q_2(1)=0$.

\subsection{The generalized Demchenko case without twisting in $\R^4$. A qualitative analysis of the solutions}

In dimension four, the equations of motion of generalized Demchenko case without twisting reduce to Hamiltonian equations with respect to the Poisson structure \eqref{poisson} on the cotangent bundle $T^*S^{3}\subset\mathbb{R}^4\{\gamma,p\}$ of the three-dimensional sphere realized by $\langle\gamma,\gamma\rangle=1$, $\langle\gamma,p\rangle=0$.
 Let
\[
\kappa=\kappa_{12} \mathbf e_1\wedge \mathbf e_2+\kappa_{34} \mathbf e_3\wedge \mathbf e_4.
\]

The equations \eqref{dem} are:
\begin{equation}\label{eq:dim4}
\begin{aligned}
{\dot\gamma}_1&=\frac{\varepsilon^2}{\tau}p_1,\qquad\qquad {\dot p}_1=\frac{1}{\tau}\kappa_{12}p_2+\mu\gamma_1,\\
{\dot\gamma}_2&=\frac{\varepsilon^2}{\tau}p_2,\qquad\qquad {\dot p}_2=-\frac{1}{\tau}\kappa_{12}p_1+\mu\gamma_2,\\
{\dot\gamma}_3&=\frac{\varepsilon^2}{\tau}p_3,\qquad\qquad {\dot p}_3=\frac{1}{\tau}\kappa_{34}p_4+\mu\gamma_3,\\
{\dot\gamma}_4&=\frac{\varepsilon^2}{\tau}p_4,\qquad\qquad {\dot p}_4=-\frac{1}{\tau}\kappa_{34}p_3+\mu\gamma_4,\\
\mu&=\frac{1}{\tau}\big(\kappa_{12}(p_1\gamma_2-p_2\gamma_1)+\kappa_{34}(p_3\gamma_4-p_4\gamma_3)\big)-\frac{\varepsilon^2}{\tau}(p_1^2+p_2^2+p_3^2+p_4^2).
\end{aligned}
\end{equation}
The Hamiltonian is
\[
h=\frac{\varepsilon^2}{ 2  \tau}(p_1^2+p_2^2+p_3^2+p_4^2).
\]

\begin{thm}\label{3sfera}
The reduced equations of generalized Demchenko case for $n=4$ \eqref{eq:dim4} are Liouville integrable on $T^*S^3$ with the three first integrals
$h$, $\Phi_{12}$, and $\Phi_{34}$ in involution, where
\begin{align*}
&\Phi_{12}(p,\gamma)=\gamma_1 p_2-\gamma_2p_1+\frac{\kappa_{12}}{2\varepsilon^2}(\gamma_1^2+\gamma_2^2), \\
&\Phi_{34}(p,\gamma)=\gamma_3 p_4-\gamma_4p_3+\frac{\kappa_{34}}{2\varepsilon^2}(\gamma_3^2+\gamma_4^2).
\end{align*}
\end{thm}

The proof  follows by a direct calculation.

It is well known that the question of integrability for a Hamiltonian system is distinct from the problem of its explicit integration.

The reduced equations of generalized Demchenko case without twisting in $\R^4$ can be solved via elliptic functions by quadratures, similarly to their three-dimensional counterpart, see Theorem \ref{th:Integration3d} above.

\begin{thm}\label{th:Integration}
The reduced equations of generalized Demchenko case without twisting for $n=4$ \eqref{eq:dim4} can be explicitly integrated via elliptic functions and their degenerations.
\end{thm}
\begin{proof}
 As in dimension $n=3$, instead on the cotangent bundle $T^*S^2\{\gamma,p\}$,  we will integrate the system
on the tangent bundle $TS^3\{\gamma,\dot\gamma\}$.
Let us introduce new coordinates $\rho_1, \rho_3, \varphi_1, \varphi_3$ by
\[
\gamma_1=\rho_1\cos\varphi_1,\quad \gamma_2=\rho_1\sin\varphi_1,\quad \gamma_3=\rho_3\cos\varphi_3,\quad \gamma_4=\rho_3\sin\varphi_3.
\]

From the condition $\langle\gamma,\gamma\rangle=1$ it follows that $\rho_1^2+\rho_3^2=1$, while $\langle\gamma, \dot\gamma \rangle=0$ is identically satisfied.
In the new coordinates the first integrals  become
\begin{equation}\label{int}
\begin{aligned}
& h=\frac{\tau}{2\varepsilon^2}\left(\dot{\rho}_1^2+\rho_1^2\dot{\varphi}_1^2+\dot{\rho}_3^2+\rho_3^2\dot{\varphi}_3^2\right),\\
& \Phi_{12}=\frac{\tau}{\varepsilon^2}\rho_1^2\dot{\varphi}_1+\frac{\kappa_{12}}{2\varepsilon^2}\rho_1^2,\\
& \Phi_{34}=\frac{\tau}{\varepsilon^2}\rho_3^2\dot{\varphi}_3+\frac{\kappa_{34}}{2\varepsilon^2}\rho_3^2.
\end{aligned}
\end{equation}

Since the first integrals $\Phi_{12}$ and $\Phi_{34}$ depend on $\rho_1, \dot{\varphi}_{1}$ and $\rho_3, \dot{\varphi}_{3}$ respectively, $\dot{\varphi}_1$ can be expressed as
a function of $\rho_1$ and values of these first integrals; similarly,  $\dot{\varphi}_3$ can be expressed as
a function of $\rho_3$ and values of these first integrals:
\begin{equation}
\label{int1}
\dot{\varphi}_1=\frac{2\varepsilon^2\Phi_{12}-\kappa_{12}\rho_1^2}{2\tau\rho_1^2},\qquad
\dot{\varphi}_3=\frac{2\varepsilon^2\Phi_{34}-\kappa_{34}\rho_3^2}{2\tau\rho_3^2}.
\end{equation}

By differentiating  the relation $\rho_1^2+\rho_3^2=1$ with respect to time,  we  get
\[
\dot{\rho}_3^2=\frac{\rho_1^2}{1-\rho_1^2}\dot{\rho}_1^2.
\]
Using \eqref{int1}, the last equality, and the expression for the first integral $h$ from \eqref{int}, one obtains
\[
\dot{\rho}_1^2=(1-\rho_1^2)\frac{2\varepsilon^2h}{\tau}-\frac{(2\varepsilon^2\Phi_{34}-\kappa_{34}+\kappa_{34}\rho_1^2)^2}{4\tau^2}-
\frac{1-\rho_1^2}{\rho_1^2}\frac{(2\varepsilon^2\Phi_{12}-\kappa_{12}\rho_1^2)^2}{4\tau^2}.
\]

Introducing $u=\rho_1^2$,  it follows
\begin{equation}\label{u}
\dot{u}^2=P_3(u).
\end{equation}

Here, $P_3$ is a polynomial in $u$ of the degree not greater than three:
\[
P_3(u):=a_0u^3+a_1u^2+a_2u+a_3,
\]
where
\begin{align*}
a_0&=\frac{\kappa_{12}^2-\kappa_{34}^2}{\tau^2},\qquad\qquad a_3=-\frac{4\varepsilon^4\Phi_{12}^2}{\tau^2},\\
a_1&=-\frac{8\varepsilon^2h}{\tau}-\frac{2\kappa_{34}}{\tau^2}(2\varepsilon^2\Phi_{34}-\kappa_{34})-\frac{\kappa_{12}^2}{\tau^2}-\frac{4\varepsilon^2\kappa_{12}\Phi_{12}}{\tau^2},\\
a_2&=\frac{8\varepsilon^2h}{\tau}-\frac{(2\varepsilon^2\Phi_{34}-
\kappa_{34})^2}{\tau^2}+\frac{4\varepsilon^2\kappa_{12}\Phi_{12}}{\tau^2}+\frac{4\varepsilon^4\Phi_{12}^2}{\tau^2}.
\end{align*}

 Therefore,  from the equation \eqref{u}, integrating, one gets $\rho_1^2$ as an elliptic function or a degeneration of an elliptic function,
depending on the degree and composition of zeros of the polynomial $P_3(u)$.  We get $\rho_3$  from the algebraic equation
$\rho_3^2=1-\rho_1^2$. Finally, the variables $\varphi_1, \varphi_3$ can be obtained by quadratures from \eqref{int1}. \end{proof}

Let us express the variable $\rho_1^2$ in terms of the Weierstrass $\wp$-function  in a generic case:
$\kappa_{12}^2\ne \kappa_{34}^2$ and  the polynomial $P_3(u)$ has all roots distinct.
Introducing $z$ such that
\[
u=\frac{4}{a_0}z-\frac{a_1}{3a_0},
\]
 the equation \eqref{u} takes the form
\begin{equation}\label{medjukorak}
\dot{z}^2=4z^3-g_2z-g_3,
\end{equation}
where
\[
g_2=\frac{a_1^2}{12}-\frac{a_0a_2}{4}, \qquad g_3=\frac{a_0a_1a_2}{4}-\frac{a_1^3}{216}-\frac{a_0^2a_3}{16}.
\]
 By integration of \eqref{medjukorak} we get
\[
\int\limits_z^\infty\frac{d\xi}{\sqrt{4\xi^3-g_2\xi-g_3}}-\int\limits_{z_0}^\infty\frac{d\xi}{\sqrt{4\xi^3-g_2\xi-g_3}}=\pm (t-t_0).
\]

 Finally,  using the Weierstrass $\wp$-function (see for example \cite{Akh4}), one obtains
\[
z=\wp(A\pm(t-t_0)),\quad z_0=\wp(A).
\]

Now, we are going to provide a qualitative analysis of the solutions of the generalized Demchenko case without twisting in $\R^4$, obtained
in Theorem \ref{th:Integration}.

{\bf Case A.} Let us consider first the case $\kappa_{12}^2\ne \kappa_{34}^2$. Then $P_3(u)$ is a degree three polynomial.
The coordinates $\rho_1, \varphi_1$ and $\rho_3, \varphi_3$ are polar coordinates on the projections of the sphere $\langle\gamma,\gamma\rangle=1$
 to  the coordinate planes $O\mathbf e_1\mathbf e_2$ and $O\mathbf e_3\mathbf e_4$, respectively.  Hence, $\rho_1$ and $\rho_3$, and consequently $u$ can take values  between $0$ and $1$.

Since
\[
P_3(0)=-\frac{4\varepsilon^4{\Phi_{12}^{2}}}{\tau^2}<0,
\]
and
\[
P_3(1)=-\frac{4\varepsilon^4{\Phi_{34}^{2}}}{\tau^2}<0,
\]
one concludes that
on interval $(0,1)$ the polynomial $P_3(u)$ has (i) no real roots; (ii) two distinct real roots; or (iii) one double real root.

\begin{itemize}

\item[(i)]
If the number of real roots is zero, then the polynomial $P_3(u)$ takes negative values on the whole interval $(0,1)$. Thus, the case (i)
does not correspond to a real motion.

\item[(ii)]
In the case (ii)
when the polynomial $P_3(u)$ has two distinct real roots $u_1<u_2$ on the interval $(0,1)$, the projection of a trajectory  to
the $O\mathbf e_1\mathbf e_2$ and $O\mathbf e_3\mathbf e_4$ planes belong, respectively, to the annuli
\[
\sqrt{u_1}\leqslant\rho_1\leqslant\sqrt{u_2} \qquad \text{and} \qquad \sqrt{1-u_2^2}\leqslant\rho_3=\sqrt{1-\rho_1^2} \leqslant\sqrt{1-u_1^2}.
\]
There are three types of the trajectories in this case.  Let
\[
\hat{u}=\frac{2\varepsilon^2\Phi_{12}}{\kappa_{34}}.
\]
 If $\hat{u}$  belongs to $(u_1,u_2)$ then $\dot{\varphi}_1$ changes the sign and trajectories are presented in Figure
\ref{fig:slika2ab}. If $\hat{u}$ is equal to $u_1$ or $u_2$, then the trajectories are presented
in Figure \ref{fig:slika3ab}. Otherwise, the trajectories are presented in Figure \ref{fig:slika4}.

\begin{figure}[h]
{\centering
{\includegraphics[width=11cm]{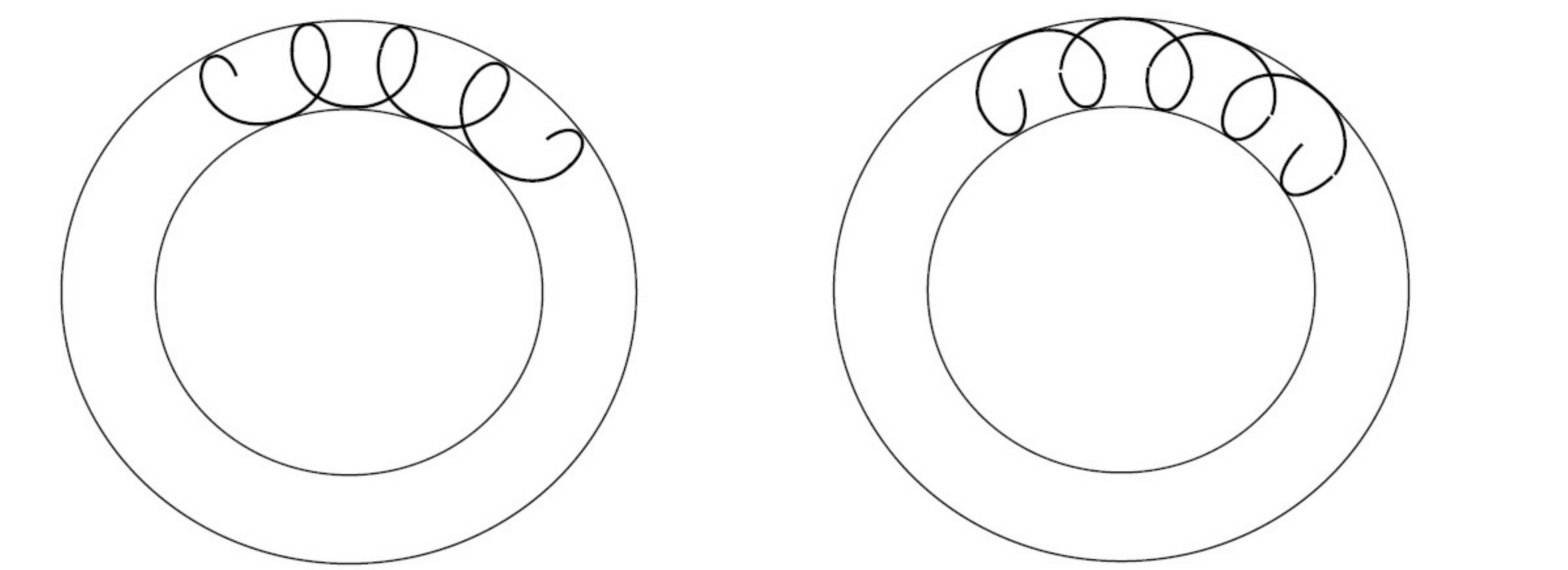}}
\caption{The case $u_1<\hat{u}<u_2$} \label{fig:slika2ab}}
\end{figure}

\begin{figure}[h]
{\centering
{\includegraphics[width=10cm]{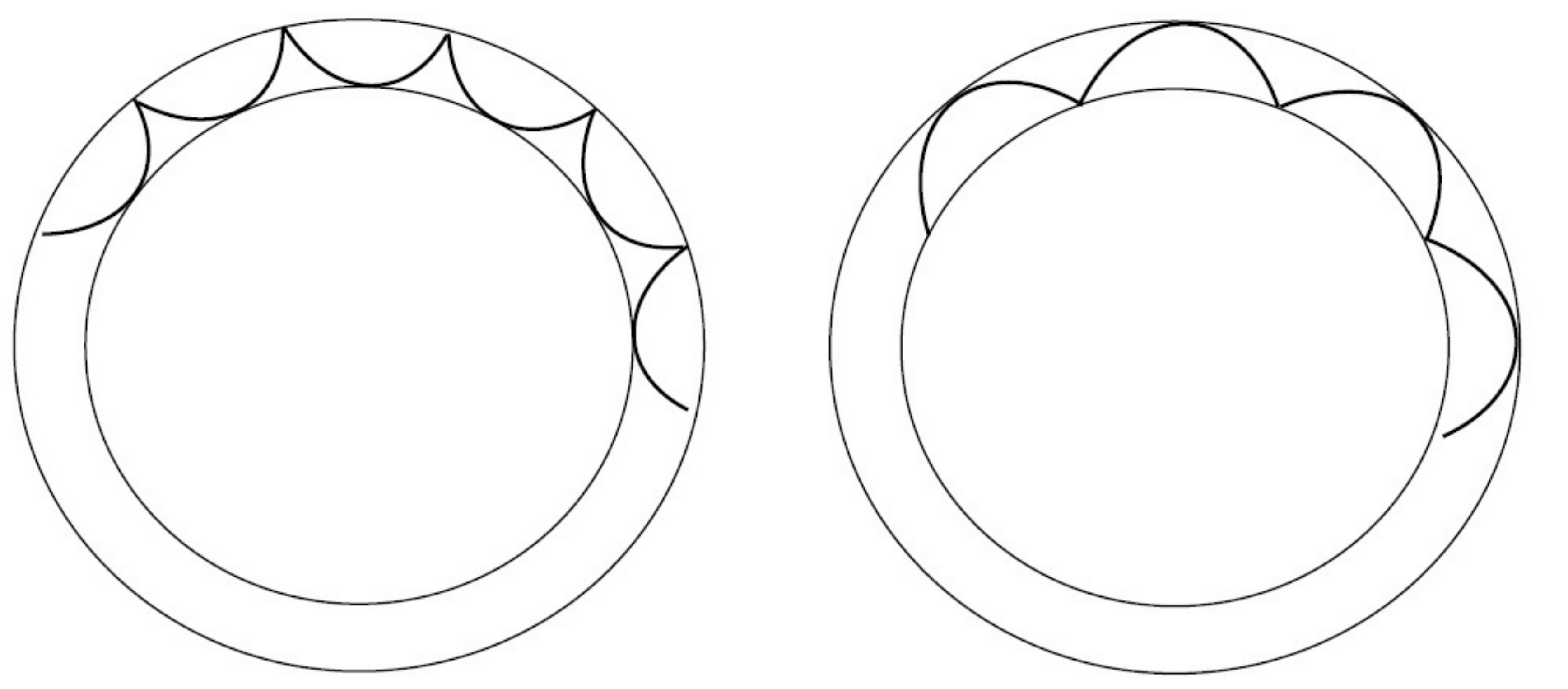}}
\caption{The cases $\hat{u}=u_2$ (left) and $\hat{u}=u_1$ (right)}\label{fig:slika3ab}}
\end{figure}

\begin{figure}[h]
{\centering
{\includegraphics[width=4.2cm]{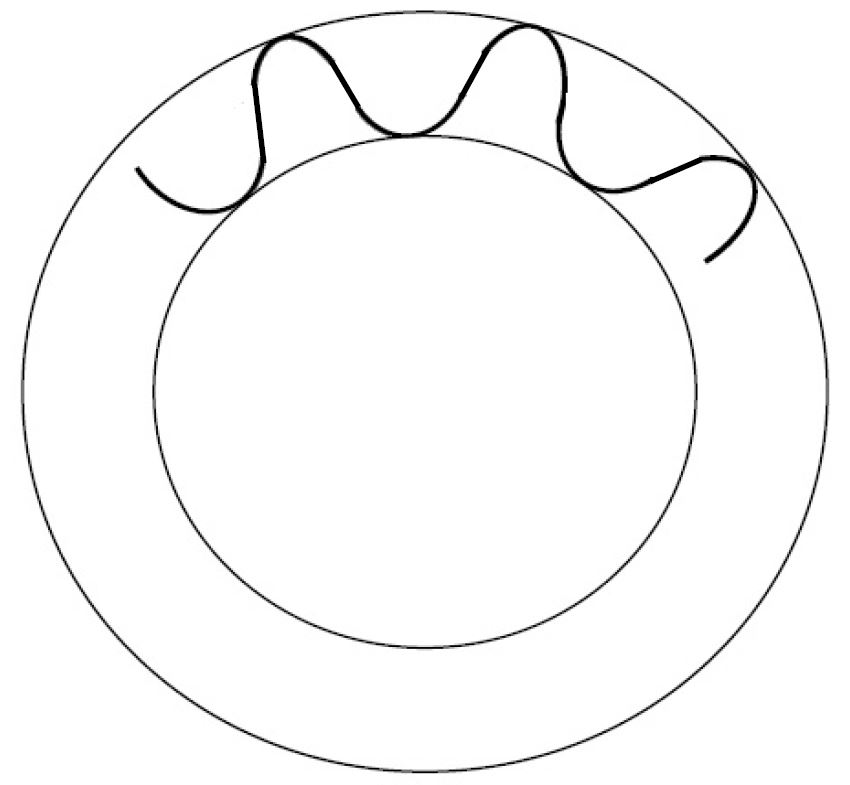}}
\caption{The case when $\hat{u}$ does not belong to the interval $[u_1,u_2]$ }\label{fig:slika4}}
\end{figure}

\begin{figure}[h]
{\centering
{\includegraphics[width=4.2cm]{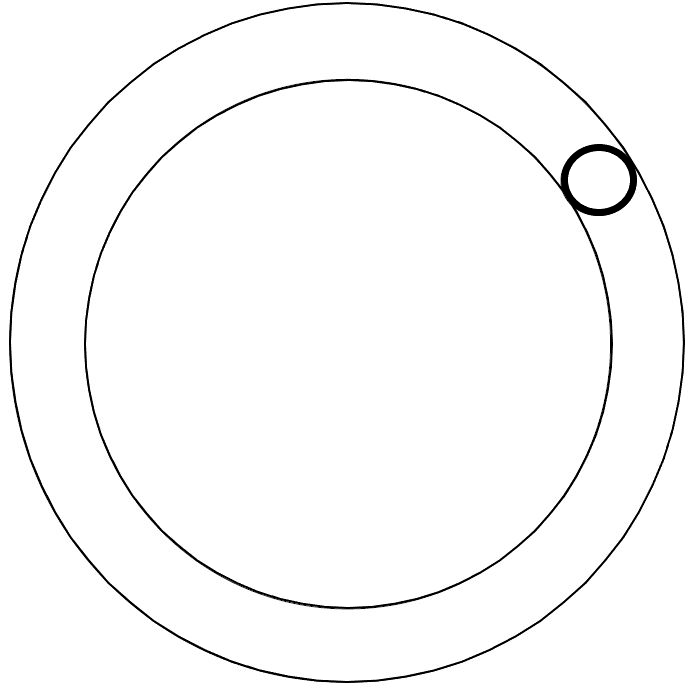}}
\caption{A case that does not correspond to a possible motion.}\label{fig:slika5}}
\end{figure}

\item[(iii)]
The case of a double root $u_1=u_2$ corresponds to the stationary motion
\begin{align*}
&\rho_1=const,    \qquad  \qquad   \qquad \varphi_{1}=\alpha_1t+\varphi_{10}, \\
&\rho_3=\sqrt{1-\rho_1^2}=const,   \quad \varphi_{3}=\alpha_3t+\varphi_{20},
\end{align*}
where
\[
\alpha_1=\frac{2\varepsilon^2\Phi_{12}-\kappa_{12}u_1}{2\tau u_1}=const, \quad
\alpha_3=\frac{2\varepsilon^2\Phi_{34}-\kappa_{34}(1-u_1)}{2\tau(1-u_1)}=const.
\]
From the equations of motion \eqref{eq:dim4} it follows that the constants $\alpha_1$ and $\alpha_3$ should satisfy:
\[
\kappa_{12}\alpha_{1}-\kappa_{34}\alpha_{3}+\tau(\alpha_1^2-\alpha_3^2)=0.
\]
Since the roots $u_1$ and $u_2$ of the polynomial $P_3(u)$ coincide, the discriminant of the polynomial $P_3(u)$ is equal to zero.

As we mentioned, in the case when $\dot{\varphi_1}$ changes the sign, the trajectories are presented in Figure \ref{fig:slika2ab}. In both cases, if we consider $\varphi_1$ as a function on the universal covering of $S^1$, it is an unbounded function of time: in one case it goes to plus infinity, while in the other case it goes to minus infinity, when $t$ goes to infinity.

We come to a natural question: \emph{is there any case when $\varphi_1$ is a bounded or, in particular, a periodic function of time?}

In other words, are there conditions which would generate  Figure \ref{fig:slika5} as a limit case of those presented in Figure \ref{fig:slika2ab}.  The answer is negative, as one concludes from the following:

\begin{prop} \label{neogranicena}
If $\kappa_{12}\ne0$, then $\varphi_1$ is unbounded function of time.
\end{prop}
\begin{proof} From \eqref{int1} we have
\[
\dot{\varphi}_1=\frac{2\varepsilon^2\Phi_{12}}{2\tau u}-\frac{\kappa_{12}}{2\tau}.
\]
Since $\kappa_{12}\ne 0$, the second  addend is a constant, while the first one is periodic in time. So $\varphi_1$ is unbounded function of time.
\end{proof}
\end{itemize}

{\bf Case B.}
In the  case  $\kappa_{34}=\pm\kappa_{12}$, the coefficient of $u^3$ in the polynomial $P_{3}(u)$ is zero. Hence $P_3(u)$ is at most a quadratic polynomial in $u$. Qualitative pictures of the trajectories are the same as before. They  are presented in Figures \ref{fig:slika2ab}, \ref{fig:slika3ab}, and \ref{fig:slika4} with an important difference: now the solutions are not elliptic functions of time.

In the case when $u_1=u_2$, the discriminant of the polynomial $P_3$ vanishes. This leads to the
stationary motion
\[
\rho_1=const, \quad \rho_3=\sqrt{1-\rho_1^2}=const, \quad \varphi_{1}=\alpha_1t+\varphi_{10}, \quad \varphi_{3}=\alpha_3t+\varphi_{20}.
\]

As in the case A, the constants $\alpha_1$ and $\alpha_3$ are not independent. If $\kappa_{12}=\kappa_{34}$ we have $\alpha_1=\alpha_3$, or $\alpha_1+\alpha_3={\kappa_{12}}/{\tau}$. When $\kappa_{34}=-\kappa_{12}$, then $\alpha_{1}=-\alpha_3$ or $\alpha_1+\alpha_3={\kappa_{12}}/{\tau}$.

\begin{rem}
 Let us remark  that in the dynamics of the Lagrange top in absence of gravity  there exist a situation
 similar to the one mentioned before Proposition \ref{neogranicena} (see Figure \ref{fig:slika5}). This system can also be seen as a symmetric Euler top.
 There is a stationary motion about the axis of symmetry that is in a non-vertical position. In other  words, the  system of equations admits the following particular solution: the nutation angle $\theta=\theta_0\in (0,\pi/2)$ is a constant different from zero, the precession angle $\varphi$ is constant, and the angle of intrinsic rotation $\psi$ is a linear function of time. If in an initial moment of time one chooses $\theta$ close to $\theta_0$, then the nutation
  and precession will be periodic functions  of time, and the axis of symmetry will uniformly rotate about the vector of angular momentum, which is fixed in the space. See \cite{Ar} for more details.

 What is going on in with the Lagrange top with the presence of  gravity? Can the precession angle be a periodic function on the universal covering of $S^1$?

 It may look like  the mentioned stationary solution exists in the presence of gravity as well. The three first integrals (the energy integral, the projection of the angular momentum on the vertical axis, the projection of the angular momentum on the axis of symmetry) are constant functions on the solution. However, from the equations of motion one gets that the stationary motion about the axis of symmetry is possible only when $\theta=0$ or $\theta=\pi$. Based on that, one can speculate that a solution of the Lagrange top with the presence of  gravity having the precession angle as a bounded or periodic function of time does not exist. A rigorous proof of that observation was provided by Hadamard in 1895 \cite{H1895}. Although the Lagrange top was  widely studied    since then, with  dozens of volumes devoted to it, this Hadamard's result is very hard to find. A nice exception is a recent short note \cite{ZS2018}.
\end{rem}

\subsection*{Acknowledgements}
We are very grateful to the referees for valuable remarks that significantly
helped us to improve the exposition, and in particular for providing us an example which led to Remark \ref{rem:ref}. We thank Viswanath Ramakrishna for reading the manuscript and providing a feedback. This research has been supported by the Project no. 7744592 MEGIC ”Integrability
and Extremal Problems in Mechanics, Geometry and Combinatorics” of the Science Fund
of Serbia, Mathematical Institute of the Serbian Academy of Sciences and Arts and the
Ministry for Education, Science, and Technological Development of Serbia, and the Simons
Foundation grant no. 854861.


\end{document}